\documentclass[a4paper,UKenglish,cleveref, autoref, thm-restate]{lipics-v2021}

\hideLIPIcs  

\usepackage{comment}
\usepackage{todonotes}
\usepackage{mathtools}
\usepackage{tikz}
\usepackage{hyperref}
\usetikzlibrary{arrows.meta,positioning,calc,fit, backgrounds,shapes.misc}

\usepackage{xspace}
\newcommand{\LET}{\textsc{Lettericity}\xspace}
\newcommand{\CE}{\textsc{Coloring Extension}\xspace}
\newcommand{\DE}{\textsc{Decoder Extension}\xspace}
\newcommand{\WE}{\textsc{Word Extension}\xspace}

\newcommand{\NAE}{\textsc{Monotone-NAE3-SAT}\xspace}

\newcommand{\Oh}{\mathcal{O}}
\usepackage[]{todonotes}

\title{Lettericity Is \textsf{NP}-Complete} 
 
\author{Henning Fernau}{Fachbereich~4 -- Abteilung Informatikwissenschaften, Universit\"at Trier, 54286 Trier, Germany}{}{https://orcid.org/0000-0002-4444-3220}{}

\author{Samuel German\footnote{Corresponding author.}}{Department of Computer Science and Engineering, University of California San Diego, 9500 Gilman Drive, Mail Code 0404, La Jolla, CA 92093-0404, USA \and \url{https://samuelgerman6.github.io/} }{sgerman@ucsd.edu}{https://orcid.org/0009-0002-5992-8391}{}

\author{Kevin Mann}{Fachbereich~4 -- Abteilung Informatikwissenschaften, Universit\"at Trier, 54286 Trier, Germany}{}{https://orcid.org/0000-0002-0880-2513}{}

\author{Nils Morawietz}{Friedrich Schiller University Jena, Institute of Computer Science, Germany}{nils.morawietz@uni-jena.de}{https://orcid.org/0000-0002-7283-4982}{}

\authorrunning{H. Fernau, S. German, K. Mann, N. Morawietz} 

\Copyright{Henning Fernau, Samuel German, Kevin Mann, and Nils Morawietz} 

\ccsdesc[100]{Theory of computation~Problems, reductions and completeness} 
\ccsdesc{Mathematics of computing~Graph theory}

\keywords{Lettericity, NP-completeness, graph coloring, Exponential Time Hypothesis} 

\category{} 

\relatedversion{} 

\acknowledgements{This work was initiated, and the substantial majority of it developed, while S.G. was visiting H.F. and K.M. at the University of Trier; S.G. thanks the University of Trier and its Theoretical Computer Science research group for their hospitality.}

\nolinenumbers 

\begin{document}

\maketitle

\begin{abstract}
The lettericity of a graph $G$ is the smallest size of a set $\Sigma$ such that there exist $w_1, \ldots, w_{|V(G)|} \in \Sigma$ and a so-called decoder $D \subseteq \Sigma^2$ for which $G$ is isomorphic to the letter graph $(\{1, \ldots, |V(G)|\}, \{ij : 1 \le i < j \le |V(G)|, w_iw_j \in D\})$. It took around two decades of the study of lettericity for, in the simpler case of paths, a closed-form expression for its lettericity to be derived; this suggests that the question of whether the lettericity of an arbitrary graph can be computed in polynomial time is nontrivial. Indeed, this question has been raised repeatedly as an open problem in recent literature. 

We solve this problem by showing that the lettericity problem on arbitrary graphs is \textsf{NP}-complete (Theorem~\ref{thm:np-complete}). We also show that each of the natural ways in which a single part of a structure $(\Sigma, w_1, \ldots, w_{|V(G)|}, D)$ witnessing the lettericity could be constrained corresponds to an \textsf{NP}-complete problem, which resolves an additional three open problems from the literature. These results are as follows. We prove that the coloring extension problem --- the same problem as lettericity, with the added condition that if $f$ is the isomorphism mapping from $G$ to the letter graph, $w_{f(v)} = \chi(v)$ for a given coloring $\chi$ of $G$   --- is \textsf{NP}-complete even when $\chi$ is proper and maps exactly two vertices to each color, which resolves another open question posed in recent literature (Theorem~\ref{thm:coloring-extension-npcomplete}). We also resolve the open problem of classifying the complexity of the word extension problem, which is the same problem as lettericity except that the $w_i$ are fixed; we show it to be \textsf{NP}-complete even when $G$ is a disjoint union of equally sized cliques (Theorem~\ref{thm:word-extension-np-complete}), which, in tandem with our \textsf{NP}-completeness result for coloring extension, contrasts with the known result that when the constraint of the coloring extension problem and the constraint of the word extension problem are both applied to lettericity, lettericity can be decided in polynomial time. 
Similarly, we sharpen the known \textsf{GI}-hardness result of the decoder extension problem, where the decoder is fixed, towards \textsf{NP}-completeness (Theorem~\ref{thm:decoder-extension-np-complete}).

Additionally, we use the reduction in the \textsf{NP}-completeness proof to give a lower bound on algorithms for lettericity under the Exponential Time Hypothesis (ETH): unless the ETH is false, there cannot exist a deterministic algorithm to decide whether the lettericity of an $n$-vertex graph is at most~$k$ in time $2^{o(n)}$, even when $n = 6k$ (Theorem~\ref{thm:eth-lowerbound}).

Finally, we complement the known result that \textsc{Lettericity} is fixed-parameter tractable when parameterized by the standard parameter~$k$ by providing a fixed-parameter algorithm with respect to the dual parameterization by $n-k$ running in $2^{\Oh((n-k)^2)} + \Oh(n)$~time (Theorem~\ref{thm:fpt-dual}).
\end{abstract}
\clearpage
\section{Introduction}\label{sec:intro}

Lettericity was introduced as a graph parameter in \cite{Pet2002}. For a graph  $G$, its \emph{lettericity} is the smallest size of a set $\Sigma$ for which  there exist $w_1, \ldots, w_{|V(G)|} \in \Sigma$ and a \emph{decoder}~$D \subseteq \Sigma^2$ such that $G$ is isomorphic to what is known as the \emph{letter graph}, namely $\Gamma_{D}(w)=(\{1, \ldots, |V(G)|\}, \{ij : 1 \le i < j \le |V(G)|, w_iw_j \in D\})$. This notion has received quite some attention in recent years, see (e.g.) \cite{AleKLZ2026,AleLozMal2024,AleVilLoz2026,AtmLoz2024,Bra2025,FenFFKS2026,FenFMRS2025,Fer2020,GroMorSac2026,Loz2023,ManVat2024}. The question of whether it would be possible to determine the lettericity of a given graph in polynomial time was repeatedly raised as an open question in recent literature, see \cite{AleKLZ2026,AleLWZ2020,FenFMRS2025,GroMorSac2026}; it was suggested as a research question already in the very first paper on this graph parameter~\cite{Pet2002}.
Only after around two decades of attention to lettericity was a closed-form expression for the lettericity of a path derived~\cite{Fer2020}, suggesting a nontrivial nature to the problem of understanding the complexity of computing the lettericity of an arbitrarily given graph. Another interesting aspect here is that related parameters can be either computable in polynomial time or \textsf{NP}-complete.
\begin{itemize}
    \item By \cite{AleLWZ2020}, the lettericity of a graph is upper-bounded by its neighborhood diversity, a parameter that can be determined in polynomial time. In fact, neighborhood diversity can be even characterized by a variant called \emph{symmetric lettericity} in \cite{GroMorSac2026}.
    \item By \cite{AleLWZ2020}, the lettericity of a graph (plus one) upper-bounds its linear clique-width, a parameter whose computation yields an \textsf{NP}-hard problem; additionally, a structure $(\Sigma, w_1, \ldots, w_{|V(G)|}, D)$ witnessing the lettericity can  be used to compute a linear clique-width expression with at most $1 + |\Sigma|$ labels~\cite[Theorem~9]{AleLWZ2020}.
    \item By \cite{FenFFKS2026}, the lettericity of a graph is also lower-bounded by its thinness, a parameter only recently proven to pose an \textsf{NP}-hard problem~\cite{Shi2025}.
\end{itemize}

A further application of lettericity is that a permutation class is geometrically griddable if and only if the inversion graphs of its permutations have bounded lettericity~\cite[Theorem~1.1]{AleFKLVZ2022}, and structures $(\Sigma, w_1, \ldots, w_{|V(G)|}, D)$ witnessing the lettericities of these inversion graphs can be used to construct geometric griddings of the permutation  class~\cite[Section~6.5]{AleFKLVZ2022}. Moreover, for any fixed MSO-definable hereditary property, if $G$ has bounded lettericity of at most~$k$, there is a $f(k)n^3$ algorithm for $n = |V(G)|$ to determine whether said property holds~\cite[Theorem~8]{AleKLZ2026}. Thus beyond providing a compact structural description of a graph, lettericity has meaningful algorithmic consequences.

The main result of this paper is to establish that the lettericity problem is in fact \textsf{NP}-complete on general graphs (Theorem~\ref{thm:np-complete}), resolving the mentioned open problem. We also show that each of the natural ways of making a problem out of fixing a single part of a structure $(\Sigma, w_1, \ldots, w_{|V(G)|}, D)$ witnessing the lettericity forms an \textsf{NP}-complete problem, which resolves additional three open problems from the literature. (1) We resolve the complexity of the coloring extension problem, which is the same problem as lettericity with the added condition that for the isomorphism mapping~$f$ from~$G$ to the letter graph, $w_{f(v)} = \chi(v)$ for a given coloring~$\chi$ of~$G$. The paper~\cite{GroMorSac2026} recently posed as an open question the complexity of coloring extension. We show that even in the restricted case where $\chi$ is a proper coloring which maps exactly 2 vertices to each color, the coloring extension problem is \textsf{NP}-complete. 
(2) We prove that the word extension problem is \textsf{NP}-complete (Theorem~\ref{thm:word-extension-np-complete}), which resolves another open problem which was asked recently~\cite{GroMorSac2026}. It is the same problem as lettericity except that each of the $w_i$ must be equal to a specific value. Note that neither the \textsf{NP}-hardness of the word extension problem nor that of the coloring extension problem follows immediately from the \textsf{NP}-hardness of lettericity, as $\exists y \ R(x,y) $ being \textsf{NP}-hard to decide for a given~$x$ does not imply that $R(x,y)$ (for given $x$ and~$y$) is \textsf{NP}-hard to decide. Given that lettericity, when both the added  constraint of the coloring extension problem and the added constraint of the word extension problem are enforced upon it, can be decided in polynomial time~\cite[Theorem~2]{GroMorSac2026}, our proofs that the coloring extension and word extension problems are \textsf{NP}-complete help clarify where the hardness of lettericity lies. (3) In Theorem~\ref{thm:decoder-extension-np-complete}, we also prove that the decoder extension problem (the same problem as lettericity, except that $\Sigma$ and $D \subseteq \Sigma^2$ must be equal to objects given as part of the input) is \textsf{NP}-complete, sharpening a \textsf{GI}-hardness observation of~\cite{GroMorSac2026}.

Although Theorem~\ref{thm:np-complete} rules out there being a polynomial-time algorithm for lettericity under the assumption $\mathsf{P} \neq \mathsf{NP}$, it does not by itself rule out a deterministic subexponential-time algorithm for lettericity under $\mathsf{P} \neq \mathsf{NP}$.  We use the reduction employed to demonstrate the \textsf{NP}-hardness of lettericity to prove that, unless the Exponential Time Hypothesis is false, there cannot exist a deterministic algorithm to decide whether the lettericity of a graph on $n$ vertices is at most $k$ in  time $2^{o(n)}$, even when $n = 6k$ (Theorem~\ref{thm:eth-lowerbound}).

Finally, we also look into parameterized algorithms. In \cite{AleKLZ2026}, it was shown that \textsc{Lettericity} is in \textsf{FPT} when parameterized by solution size~$k$.
We 
complement this in Theorem~\ref{thm:fpt-dual} by proving that there is an \textsf{FPT} algorithm for \textsc{Lettericity} parameterized by the dual parameter $n-k$ running in $2^{\mathcal{O}((n-k)^2)} + \mathcal{O}(n)$ time. This is a phenomenon rarely observed for parameters on general graphs.

\smallskip
\noindent
\emph{Independent work.}  Shortly after we had posted a preprint of this paper, Vatter~\cite{Vat2026} posted a different proof that \textsc{Lettericity} is \textsf{NP}-complete, underpinning the interest in this problem. While our  \textsf{NP}-hardness proof reduces from \textsc{Monotone-NAE3-SAT}, Vatter \cite{Vat2026} reduces from the problem of deciding whether a bipartite graph has a $P_4$ decomposition. Vatter's proof shows \textsf{NP}-hardness also for four-colorable comparability graphs. Due to the sequence of reductions from \textsc{3-SAT} over \textsc{3DM} and \textsc{Exact-Cover-By-3-Sets} to the mentioned  \textsc{$P_4$-Decomposition} problem and then to \textsc{Lettericity}, see \cite{GarJoh79,TeyRap2008,Vat2026}, the final graph instance is not at all linear in the number of variables and clauses of the original  \textsc{3-SAT} instance and is hence not useful for good ETH-based lower bounds, as opposed to ours.

\section{Preliminaries}

\subsection{General Notions}\label{sec:union}
By~$X \triangle Y = (X \setminus Y) \cup (Y \setminus X)$ we denote the symmetric difference of the sets~$X$ and~$Y$. 
For a linear order $\prec$, and distinct elements $u,v$, the \emph{open interval} between them is \[I_{\prec}(u,v) \coloneqq \{x: \min_{\prec} \{u,v\} \prec x \prec \max_\prec\{u,v\} \}.\]
Let $G_0 = (V_0, E_0)$ and $G_1 = (V_1, E_1)$ be two graphs. Their \emph{union} is $G_0 \cup G_1 = (V_0 \cup V_1, E_0 \cup E_1)$.
When the vertex sets of $G_0$ and $G_1$ are pairwise disjoint, we write $G_0 \sqcup G_1$ to denote $G_0 \cup G_1$. If $G$ is a graph or digraph and $S \subseteq V(G)$, then $G[S]$ is the subgraph \emph{induced} by~$S$, i.e., the graph $G[S]=(S,\{uv \in E(G) : \ u,v \in  S\})$; moreover,  $G-S = G[V(G) \setminus S]$. If $A,B \subseteq V(G)$ and $A \cap B = \emptyset$, then $G[A,B]$ is the bipartite graph on vertex set $A \cup B$ whose edge set consists of exactly the edges in $G[A \cup B]$ which are between $A$ and $B$. $N_G(v)$ is the \emph{open neighborhood} of vertex~$v$ (the set of vertices adjacent to $v$).

\textsc{Monotone-NAE3-SAT} takes as input a monotone CNF formula in which each clause contains exactly 3 distinct variables and asks whether there exists an assignment to its variables which makes every clause contain at least one false and one true variable. It is known  to be \textsf{NP}-complete; see \cite{DarDoc2020}.

\subsection{Notation for the Lettericity Problem}

This subsection gives an equivalent formulation of the lettericity problem which uses a function to assign letters to vertices so that the problem is expressed more similarly to the coloring problem in which there is a function assigning colors to vertices.  
A \emph{letter representation} of $G = (V,E)$ is a triple $(\prec, \lambda, D)$ where $\prec$ is a linear order of $V$, $\lambda : V \to \Sigma$ assigns a letter to every vertex, and $D \subseteq \Sigma^2$ is a decoder, such that for all $u \prec v$, \(uv \in E \iff (\lambda(u), \lambda(v)) \in D.\) $\Sigma$ is called the alphabet of the letter representation, and its elements are accordingly called letters. The \emph{lettericity} of~$G$, denoted $\ell(G)$, is the minimum value of $|\Sigma|$ in any letter representation of~$G$. This notion was introduced in \cite{Pet2002}.
\textsc{Lettericity} is the problem of deciding whether, for a  given graph and integer~$k$, $\ell(G) \le k$. 

\subsection{Coloring Extension}\label{subsec:color}
The problem formulated in this subsection is equivalent to the coloring extension problem described in 
\cite{GroMorSac2026}, but uses our lettericity notation. Let $\chi: V(G) \to [k]$ be a graph coloring (it need not be a proper graph coloring in the sense of assigning different colors to the endpoints of an edge --- here a coloring is just a vertex-to-letter map). We use the term \emph{color class} to refer to the set of vertices mapped to a certain color by $\chi$. We can always assume that $\chi$ and $\lambda$ are surjective, as unused colors can be discarded and the remaining colors renumbered; likewise, unused letters can be deleted from the alphabet and the decoder. 

As in \cite{GroMorSac2026}, 
we define what it means for a letter representation to respect~$\chi: V(G) \to [k]$: $(\prec, \lambda : V(G) \to \Sigma, D)$ \emph{respects}~$\chi$ if there exist pairwise distinct letters $\beta_1, \ldots, \beta_k \in \Sigma$ with $\lambda(v) = \beta_{\chi(v)}$ for all $v\in V(G)$; i.e., after changing $\chi$ so that color $i$ has been renamed to~$\beta_i$, 
$\lambda = \chi$. If there exists a letter representation of~$G$ respecting~$\chi$, then there is also a letter representation of~$G$ whose letter function $\lambda$ is~$\chi$: we can relabel the decoder entries according to $\beta_i \mapsto i$ when using~$\chi$ as the letter function. \CE is the problem to decide, given~$G$ and a coloring $\chi$ thereof, whether~$G$ has a letter representation  respecting~$\chi$. 

\subsection{Distinguishers and Intervals}

For a graph $G$ and two distinct vertices of it $u,v$, we define the set of \emph{distinguishers} \(\Delta_G(u,v) = (N_G(u) \triangle N_G(v)) \setminus \{u,v\}\). Thus $\Delta_G(u,v)$ is the set of all vertices in $V(G) \setminus \{u,v\}$ which are adjacent to exactly one of $u$ and $v$.

\label{sec:prelimobs}
\begin{lemma}\label{lem:orderofdistinguishers}
Let $G$ be a graph, $\chi : V(G) \to [k]$ a coloring of $G$, and $\prec$ a linear order of $G$. There exists a decoder $D$ for which $(\prec, \chi, D)$ is a letter representation of $G$ if and only if \begin{equation}\label{eq:orderofdistinguishers}\chi(u) = \chi(v) \implies \Delta_G(u,v) \subseteq I_{\prec}(u,v) \qquad (\forall u,v \in V(G) \ u \neq v). \end{equation}
In $(\prec, \chi, D)$, $\chi$ is functioning as $\lambda$ in the formulation  of a letter representation. 
\end{lemma}
\begin{proof}
The forward direction of the iff statement in Lemma~\ref{lem:orderofdistinguishers} is exactly~\cite[Lemma~1]{Fer2020}.

For the backward direction, assume \eqref{eq:orderofdistinguishers}. We need to show that there is some decoder $D$ for which $(\prec, \chi, D)$ is a letter representation of $G$. This reduces to checking whether we can define a $D$ such that for any two vertices $u,v \in V(G)$ with $u \prec v$, $uv \in E(G)$ if and only if $(\chi(u), \chi(v)) \in D$. We show that we can define such a $D$ by arguing, for every pair of colors $a,b \in [k]$, that all vertex pairs $u, v \in V(G)$ with $u \prec v$ and $\chi(u) = a$, $\chi(v) = b$ have the same status as to whether they belong to $E(G)$, which implies that a desired $D$ can be defined by including $(a,b)$ in it depending on what the exact adjacency status is for the specific color pair $a,b$. That is, we will be able to set $D\coloneqq\{(\chi(u), \chi(v)):u \prec v \ \text{and} \ uv \in E(G) \}$.

Let $a,b \in [k]$ be a color pair.  Let $u$ and $v$ be distinct vertices, and assume that $u \prec v$.  Assume that $\chi(u) = a$ and $\chi(v) = b$. Let $p$ be the first vertex, according to $\prec$, in $\chi^{-1}(a)$, and let $q$ be the last vertex, according to $\prec$, in $\chi^{-1}(b)$. Thus we have $p \preceq u \prec v \preceq q$. Because $p$ and $u$ have the same color and $v$ follows both,
\(
  pv\in E(G)\iff uv\in E(G);
\)
this is immediate if $p=u$, and otherwise follows from
\eqref{eq:orderofdistinguishers}.

Similarly, because $v$ and $q$ have the same color and $p$ precedes
both,
\(
  pv\in E(G)\iff pq\in E(G);
\)
this is immediate if $v=q$, and otherwise follows from
\eqref{eq:orderofdistinguishers}. Combining the preceding equivalences gives
\(
  uv\in E(G)
 \iff
  pq\in E(G)\).
Since $u\prec v$ was arbitrary among the pairs whose ordered colors
are $(a,b)$, all such pairs have the same adjacency status.  We may
therefore define
\(
  (a,b)\in D \iff
  pq\in E(G).
\) \end{proof}

\section{Reduction from \textsc{Monotone-NAE3-SAT} to a Specialized 
Coloring Extension Problem}\label{sec:sattocolor}
\begin{theorem}\label{thm:sattocolor}
\CE is \textsf{NP}-hard with the restriction that the given coloring is proper and every color class in the given coloring consists of exactly two vertices.
\end{theorem}

 The restrictions which Theorem~\ref{thm:sattocolor} imposes on \CE are stronger than those imposed on  \CE in Theorem~\ref{thm:colortoletter}, because in Theorem~\ref{thm:sattocolor} every color class has exactly two vertices.
 The rest of this section is devoted to proving Theorem~\ref{thm:sattocolor}. First we define some notions and gadgets requisite for the eventual proof, and show how they work. Section~\ref{sec:reduction} contains the actual construction used in the reduction. 

\subsection{Pair Classes}\label{sec:pair-classes} We use the term \emph{pair class} to refer to a set of two elements where one of the elements is labeled with a superscript of~$0$ and the other element is labeled with a superscript of~$1$. Throughout this reduction, every undirected graph considered has its vertices be elements of certain pair classes, and every directed graph considered has its vertices be certain pair classes themselves. The idea will be to have the way the $0$-vertex and $1$-vertex in a pair class compare according to $\prec$ --- where $\prec$ is the order component of a color-respecting letter representation of the constructed \CE
instance --- encode the truth value of certain relevant aspects, including those pertaining to the source  instance~$I$ of \textsc{Monotone-NAE3-SAT}. 

We use a function under the name $\tau$ as a tool in the analysis, which ultimately will be used to express what the orders for each of the pair classes should be in order for $I$ to be satisfied. $\tau$ takes as input a pair class $P_i = \{p_i^0, p^1_i \}$ and outputs $0$ or $1$. For an undirected graph $G$ whose vertex set is the union of certain pair classes, we interpret $\tau(P_i) = 1$ for $P_i \subseteq V(G)$ as $\tau$ prescribing that $p^1_i$ should succeed $p^0_i$ in any linear-order component of a letter representation of~$G$. For a letter representation $(\prec, \lambda, D)$ of $G$ to be consistent with the orientations prescribed by $\tau$ means that, for every pair class $P=\{p^0, p^1\}$ which is a subset of $V(G)$, $p^0 \prec p^1 \iff \tau(P)=1$. In the analysis, upon being  given such a function $\tau$, we consider the question of whether there exists a letter representation $(\prec, \lambda, D)$ of $G$ which is consistent with the orientations prescribed by $\tau$.

\subsection{Digraphs $L_\tau(G)$ and $R_\tau(G)$}\label{sec:support1reduction}

 Let $G$ be a graph and $\chi: V(G) \to \left[\frac {|V(G)|}  2\right]$ a coloring for which, as in the instances of \CE to which we reduce, there are 2 vertices in each color class (thus $\frac {|V(G)|}  2$ is an integer). Let $\mathcal P$ be  the pair classes of $G$ under $\chi$, and for every $P \in \mathcal P$, fix an arbitrary superscripting scheme of its vertices, viz. $P = \{p^0,p^1 \}$.

Let $\tau:  \mathcal P\to \{0,1\}$ be a function on the domain of all pair classes of $G$, mapping to~$0$ or~$1$, which we interpret, as discussed in Section~\ref{sec:pair-classes}, as prescribing a certain orientation for each pair class. 
Define the left and right endpoints of pair class $P$ according to $\tau$ as \[\ell_\tau(P) = \begin{cases} p^0, & \tau(P) = 1, \\ p^1,&\tau(P) = 0, \end{cases} \qquad  r_\tau(P) = \begin{cases} p^1, & \tau(P) = 1, \\ p^0,&\tau(P) = 0. \end{cases}\]

Define the vertex set of digraphs $L_\tau(G)$ and $R_\tau(G)$ to be~$\mathcal P$.

\noindent
\emph{Edges of $L_\tau(G)$ and $R_\tau(G)$}. If $P, Q \in \mathcal P$ for $P = \{p^0, p^1\}$, then their edges are defined as follows:
 \[P \to Q \text{ in } L_\tau(G) \iff \ell_\tau(Q) \in\Delta_G(p^0, p^1),\] \[Q \to P\text{ in } R_\tau(G)\iff r_\tau(Q) \in\Delta_G(p^0, p^1).\]

The left and right endpoints, as well as the digraphs here, will be used later throughout the reduction. A supplementary illustration of $L_\tau(G)$ on an example $G, \chi$, and $\tau$ is provided by Figure~\ref{fig:leftdigraph} in Section~\ref{sec:extra-figures} of the Appendix.

\begin{lemma}\label{lem:respecttau}
There is a letter representation of $G$ which respects the coloring $\chi$ and is consistent with the orientations prescribed
by $\tau$ if and only if both $L_\tau(G)$ and $R_\tau(G)$ are acyclic. 
\end{lemma}

\begin{proof}
Assume that there is a letter representation of $G$ which respects the coloring $\chi$ and is consistent with the orientations prescribed by $\tau$. We can assume that in this letter representation, the letter function $\lambda$ is $\chi$, as reasoned in Section~\ref{subsec:color}. Let $P \to Q$ be an arc in $L_\tau(G)$. Then $\ell_\tau(Q)$ distinguishes the two vertices of $P$, and because the vertices of $P$ have the same color, Lemma~\ref{lem:orderofdistinguishers} gives $\ell_\tau(P)\prec\ell_\tau(Q)$. It follows that, when the pair classes are ordered according to the
positions of their left endpoints in $\prec$, every arc of
$L_\tau(G)$ is directed from an earlier pair class to a later pair
class.  Therefore $L_\tau(G)$ is acyclic.
A symmetrical form of reasoning shows that $R_\tau(G)$ is acyclic. 

Conversely, assume that both digraphs are acyclic. Let $m = |\mathcal P|$. Let $P_1^L,P_2^L,\ldots,P_m^L$ be a topological ordering of $L_\tau(G)$, and let $P_1^R,P_2^R,\ldots,P_m^R$ 
be a topological ordering of $R_\tau(G)$. Consider the linear order $\prec$ of $V(G)$ with \begin{equation}\label{eq:top} 
\ell_\tau(P_1^L)
  \prec
  \ell_\tau(P_2^L)
  \prec\cdots\prec
  \ell_\tau(P_m^L)
  \prec
  r_\tau(P_1^R)
  \prec
  r_\tau(P_2^R)
  \prec\cdots\prec
  r_\tau(P_m^R). 
\end{equation}

The order~$\prec$ realizes~$\tau$, because every left endpoint precedes every right endpoint in it. We want to show that \eqref{eq:orderofdistinguishers} holds of $G$, $\chi$, and $\prec$ here in order to conclude that $(\prec,\chi, D)$ is a letter representation of $G$, which would finish proving Lemma~\ref{lem:respecttau}. To show \eqref{eq:orderofdistinguishers}, suppose $x \in \Delta_G(\ell_\tau(P), r_\tau(P))$ for some color class $P$. We have $\chi(\ell_\tau(P)) = \chi(r_\tau(P))$. By the definition of the distinguishers $\Delta_G(\cdot)$, we have $x \notin P$. Thus $x \in Q$ for some $Q \in \mathcal P\setminus \{P\}$. If $x = \ell_\tau(Q)$, then $P \to Q$ is an arc of $L_\tau(G)$, so the topological order must give $\ell_\tau(P) \prec \ell_\tau(Q)$. By construction in \eqref{eq:top}, we have $\ell_\tau(Q) \prec r_\tau(P)$ and thus $\ell_\tau(P) \prec x \prec r_\tau(P)$ so $x \in I_\prec(\ell_\tau(P), r_\tau(P))$. In the case of $x = r_\tau(Q)$, a symmetrical form of reasoning with respect to $R_\tau(G)$ yields $x \in I_\prec(\ell_\tau(P), r_\tau(P))$. Thus Condition \eqref{eq:orderofdistinguishers} is satisfied and Lemma~\ref{lem:orderofdistinguishers} gives the existence of the desired letter representation. \end{proof}

\subsection{Constructed Gadgets}\label{sec:gadgets}
Each of the gadgets defined in this subsection takes pair classes as input and returns a graph whose vertex set consists of exactly all the vertices in the pair classes said gadget takes as input. For each gadget, there is an associated coloring with the graph output by it: the coloring whose color classes are exactly the pair classes of the input to the gadget.

\begin{remark}
In the analyses in this section subsequent to the introduction of certain gadgets, we at times use the following immediate observation: if $P$ and $Q$ are distinct pair classes, then all arcs of $L_\tau(G)$ and $R_\tau(G)$ whose endpoints are $P$ and $Q$ are determined entirely by the induced subgraph $G[P\cup Q]$, because the definitions of those arcs depend only on whether a vertex of one pair class distinguishes the two vertices of the other.
\end{remark}

\label{sec:eqgadget}
\noindent\textbf{Equality Gadget.}
Let $A = \{a^0, a^1\}$ and $B = \{ b^0, b^1\}$ be two pair classes.  
The equality gadget $E(A,B)$ is the graph whose vertex set is $A \cup B$ and whose edge set consists of only the edge $a^1b^0$.
Let $\chi$ be the coloring of the graph $E(A,B)$ whose color classes are $A$ and $B$. A supplementary illustration  of $E(A,B)$
 is depicted by Figure~\ref{fig:wires} in Section~\ref{sec:extra-figures} of the Appendix. 
\label{sec:cntngadget}
\smallskip
\noindent\textbf{Containment Gadget.}
Let $A = \{a^0, a^1\}$ and $B = \{ b^0, b^1\}$ be two pair classes.  $C(A,B)$ is the graph whose vertex set is $A \cup B$ and whose edge set consists of only  $a^1b^0$ and $a^1b^1$. Let $\chi$ be the coloring of the graph $C(A,B)$ whose color classes are $A$ and $B$. 
An auxiliary illustration  of $C(A,B)$
 is depicted by Figure~\ref{fig:wires} in Section~\ref{sec:extra-figures} of the Appendix. 
\begin{lemma}[Arcs in Digraphs for Equality and Containment Gadgets]\label{lem:basic-gadget-arcs}
Let $A = \{a^0, a^1\}$ and $B = \{b^0, b^1\}$  be pair classes. Let $\tau:\{A,B\}\to\{0,1\}$ be an orientation assignment. The arc sets of $L_\tau(E(A,B))$ and $R_\tau(E(A,B))$ are given by the following table:
\[
\begin{array}{c|c|c}
(\tau(A),\tau(B))
  & L_\tau(E(A,B))
  & R_\tau(E(A,B)) \\
\hline
(1,1) & \{A\to B\} & \{A\to B\} \\
(0,0) & \{B\to A\} & \{B\to A\} \\
(0,1) & \{A\to B,\ B\to A\} & \varnothing \\
(1,0) & \varnothing & \{A\to B,\ B\to A\}
\end{array}
\]
Moreover, for the containment gadget, the only arc of $L_\tau(C(A,B))$ is $A\to B$, and the only arc of $R_\tau(C(A,B))$ is $B\to A$.
\end{lemma}

\noindent
Section~\ref{sec:arc-table-proof} of the Appendix performs the checks which show Lemma~\ref{lem:basic-gadget-arcs}.

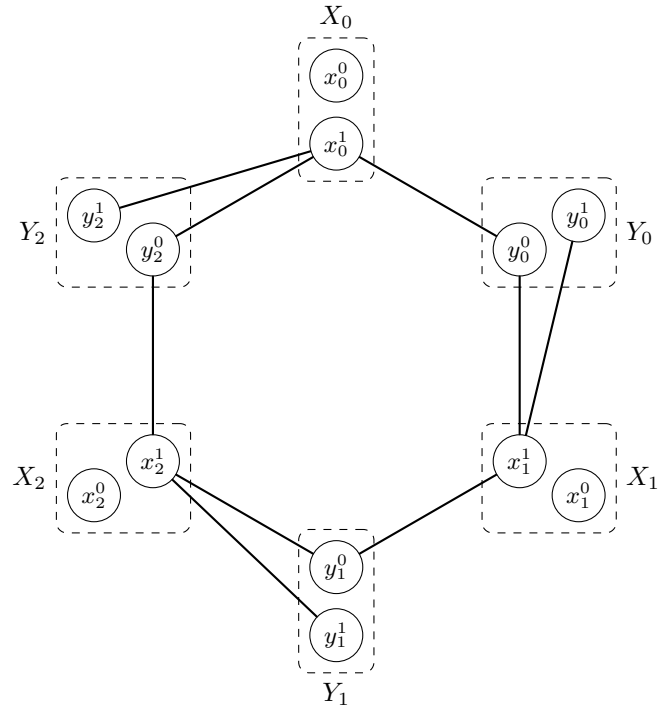
\begin{figure}[ht]
  \centering
  \begin{tikzpicture}[
      vertex/.style={
        circle,
        draw,
        fill=white,
        minimum size=7mm,
        inner sep=0pt,
        font=\small
      },
      pairclass/.style={
        draw,
        dashed,
        rounded corners,
        inner sep=4pt
      },
      graphedge/.style={
        thick,
        line cap=round
      }, scale = .7, transform shape
    ]

    %
    %

    \node[vertex] (x00) at (90:3.7) {$x_0^0$};
    \node[vertex] (x01) at (90:2.8) {$x_0^1$};

    \node[vertex] (y00) at (30:2.8) {$y_0^0$};
    \node[vertex] (y01) at (30:3.7) {$y_0^1$};

    \node[vertex] (x10) at (-30:3.7) {$x_1^0$};
    \node[vertex] (x11) at (-30:2.8) {$x_1^1$};

    \node[vertex] (y10) at (-90:2.8) {$y_1^0$};
    \node[vertex] (y11) at (-90:3.7) {$y_1^1$};

    \node[vertex] (x20) at (-150:3.7) {$x_2^0$};
    \node[vertex] (x21) at (-150:2.8) {$x_2^1$};

    \node[vertex] (y20) at (150:2.8) {$y_2^0$};
    \node[vertex] (y21) at (150:3.7) {$y_2^1$};

    \draw[graphedge] (x01) -- (y00);
    \draw[graphedge] (x11) -- (y10);
    \draw[graphedge] (x21) -- (y20);

    \draw[graphedge] (x01) -- (y20);
    \draw[graphedge] (x01) -- (y21);

    \draw[graphedge] (x11) -- (y00);
    \draw[graphedge] (x11) -- (y01);

    \draw[graphedge] (x21) -- (y10);
    \draw[graphedge] (x21) -- (y11);

    \node[
      pairclass,
      fit=(x00)(x01),
      label=above:$X_0$
    ] {};

    \node[
      pairclass,
      fit=(y00)(y01),
      label=right:$Y_0$
    ] {};

    \node[
      pairclass,
      fit=(x10)(x11),
      label=right:$X_1$
    ] {};

    \node[
      pairclass,
      fit=(y10)(y11),
      label=below:$Y_1$
    ] {};

    \node[
      pairclass,
      fit=(x20)(x21),
      label=left:$X_2$
    ] {};

    \node[
      pairclass,
      fit=(y20)(y21),
      label=left:$Y_2$
    ] {};

  \end{tikzpicture}

  \caption{The NAE clause gadget $N(\mathcal P)$. Dashed boxes
  indicate the prescribed pair classes, and all subscripts are
  interpreted modulo $3$.}
  \label{fig:nae-clause-gadget}
\end{figure}
\label{sec:naegadget}
\noindent\textbf{NAE Clause Gadget.}
For each $i \in \mathbb{Z}_3$, let $X_i = \{ x^0_i, x^1_i \}$ and $Y_i = \{y^0_i, y^1_i \}$ be pair classes. Any computation with~$i$ is considered in mod~$3$, so if for example $i = 0$, then $i -1 = 2$.

Let $\mathcal P \coloneqq \{ X_i, Y_i : i \in \mathbb{Z}_3 \}$. The NAE clause gadget on $\mathcal P$ is defined to be the graph \[N(\mathcal P) = \bigcup_{i \in \mathbb Z_3} (E(X_i, Y_i) \cup C(X_i, Y_{i-1})).\]

This definition uses the union operation on graphs, which was defined in Section~\ref{sec:union}. In particular, even though  vertices in the pair class $X_i$ are present in both of the graphs $E(X_i, Y_i)$ and $C(X_i, Y_{i-1})$, they only appear once in $N(\mathcal P)$. Thus $V(N(\mathcal P))$ is exactly the vertices in each of the 6 pair classes in $\mathcal P$, and \begin{equation}\label{eq:cledges}E(N(\mathcal P)) = \{ x^1_iy^0_i, x_i^1y^0_{i-1}, x^1_iy^1_{i-1} : i \in \mathbb{Z}_3 \}. \end{equation}
Let $\chi$ be the coloring of $N(\mathcal P)$  whose color classes are $\mathcal P$.
Figure~\ref{fig:nae-clause-gadget} illustrates 
$N(\mathcal P)$.

\begin{lemma}\label{lem:naegadget} Let $X_i = \{ x^0_i, x^1_i \}$ and $Y_i = \{y^0_i, y^1_i \}$ be pair classes, and let $\mathcal P =  \{ X_i, Y_i : i \in \mathbb{Z}_3 \}$.
Let $\tau:\mathcal P\to\{0,1\}$ be any orientation assignment. Both $L_\tau(N(\mathcal P))$ and $R_\tau(N(\mathcal P))$ are acyclic if and only if
\(
\tau(X_i)=\tau(Y_i)\) for all \(i\in\mathbb Z_3\) and
\(
(\tau(X_0),\tau(X_1),\tau(X_2))
\notin\{(0,0,0),(1,1,1)\}.
\)
\end{lemma}

\begin{proof}
We first show that if both $L_\tau(N(\mathcal P))$ and $R_\tau(N(\mathcal P))$ are acyclic, then $\tau(X_i) = \tau(Y_i)$ for all $i\in\mathbb Z_3$. Assume the antecedent of this if statement, and for the sake of contradiction, assume that $\tau(X_i)\neq\tau(Y_i)$ for some $i\in\mathbb Z_3$. Since
$N(\mathcal P)[X_i\cup Y_i]=E(X_i,Y_i)$,
the table for the arc set of the equality gadget in Lemma~\ref{lem:basic-gadget-arcs} gives a directed $2$-cycle in
$L_\tau(N(\mathcal P))$ when
$(\tau(X_i),\tau(Y_i))=(0,1)$, and a directed $2$-cycle in
$R_\tau(N(\mathcal P))$ when
$(\tau(X_i),\tau(Y_i))=(1,0)$, contradicting the acyclicity of $L_\tau(N(\mathcal P))$ and $R_\tau(N(\mathcal P))$.

Now we show that if $\tau(X_i) = \tau(Y_i)$ for all $i\in\mathbb Z_3$, then  $L_\tau(N(\mathcal P))$ and $R_\tau(N(\mathcal P))$ are acyclic if and only if  \(
(\tau(X_0),\tau(X_1),\tau(X_2))
\notin\{(0,0,0),(1,1,1)\}.
\) In tandem with what has been shown earlier in this proof, we would have a sufficiently powerful logical tool to conclude Lemma~\ref{lem:naegadget}. Assume that $\tau(X_i) = \tau(Y_i)$ for all $i\in\mathbb Z_3$. For every
$i\in\mathbb Z_3$, by construction we have 
\(
N(\mathcal P)[X_i\cup Y_{i-1}]=C(X_i,Y_{i-1})\) and \(N(\mathcal P)[X_i\cup Y_i]=E(X_i,Y_i) \).
By Lemma~\ref{lem:basic-gadget-arcs}, the only arc of
$L_\tau(N(\mathcal P))$ whose endpoints are $X_i$ and $Y_{i-1}$ is
$X_i\to Y_{i-1}$, and the only arc of $R_\tau(N(\mathcal P))$ whose
endpoints are $X_i$ and $Y_{i-1}$ is $Y_{i-1}\to X_i$. The equality
gadget $E(X_i,Y_i)$ contributes
\begin{equation}\label{eq:casearc}
\begin{cases}
X_i\to Y_i \text{ in both $L_\tau(N(\mathcal P))$ and $R_\tau(N(\mathcal P))$},&\tau(X_i)=1,\\
Y_i\to X_i \text{ in both $L_\tau(N(\mathcal P))$ and $R_\tau(N(\mathcal P))$},&\tau(X_i)=0.
\end{cases}
\end{equation}

There are no other arcs in either $L_\tau(N(\mathcal P))$ or $R_\tau(N(\mathcal P))$: if two pair classes are not of the form
$X_i,Y_i$ or $X_i,Y_{i-1}$, then by~\eqref{eq:cledges} there are no
edges between their vertices, so neither pair class can distinguish the
vertices of the other. Hence, the underlying undirected graph of each of $L_\tau(N(\mathcal P))$ and $R_\tau(N(\mathcal P))$ is exactly the cycle
\begin{equation}\label{eq:cycle}
X_0-Y_0-X_1-Y_1-X_2-Y_2-X_0.
\end{equation}

The arcs contributed to $L_\tau(N(\mathcal P))$ by the containment gadgets
are $X_i\to Y_{i-1}$ for $i\in\mathbb Z_3$. Thus this digraph contains a
directed cycle exactly when every equality-gadget arc in
\eqref{eq:casearc} is directed as $Y_i\to X_i$, which by \eqref{eq:casearc} gives the condition 
\begin{equation}\label{eq:lcyclecond}
L_\tau(N(\mathcal P))\text{ is cyclic}
\iff
\tau(X_0)=\tau(X_1)=\tau(X_2)=0.
\end{equation}
Similarly,
\begin{equation}\label{eq:rcyclecond}
R_\tau(N(\mathcal P))\text{ is cyclic}
\iff
\tau(X_0)=\tau(X_1)=\tau(X_2)=1.
\end{equation}

Thus under the assumption that  $\tau(X_i)=\tau(Y_i)$ for all $i \in \mathbb Z_3$, $L_\tau(N(\mathcal P))$ and $R_\tau(N(\mathcal P))$  are acyclic if and only if $\tau(X_0),\tau(X_1),\tau(X_2)$ are not all equal. As mentioned earlier in this proof, this proves Lemma~\ref{lem:naegadget}.
\end{proof}

\subsection{The Reduction Itself}\label{sec:reduction}
We now have the machinery we need to state and prove the overall reduction from \textsc{Monotone-NAE3-SAT} to \CE where every color class has two vertices and the coloring is proper. Consider an instance~$I$ of \textsc{Monotone-NAE3-SAT} with a set of variables~$Z$ and set of clauses~$\mathcal C$, such that $Z$ is the set of exactly those variables occurring in these clauses. 

If $I$ has no clauses, we output the edgeless graph on  two vertices, and output the coloring which assigns both of these vertices $1$. Otherwise, we apply the following construction, which involves defining $G_I$ and $\chi_I$ in terms of the gadgets previously defined. Figure~\ref{fig:reduction} gives an illustration of $G_I$ and $\chi_I$ for an example \textsc{Monotone-NAE3-SAT} instance under the gadgetized viewpoint that will help our analysis. For each $z \in Z$, create the pair class \(M_z = \{m^0_z, m_z^1 \}.\) 
For each $C \in \mathcal C$ and $i \in \mathbb Z_3$, create the pair classes \(Y_{C,i}=\{y_{C,i}^0,y_{C,i}^1\} \text{ and }  X_{C,i}=\{x_{C,i}^0,x_{C,i}^1\}.\)
All pair classes $M_z$, $X_{C,i}$, and $Y_{C,i}$ are pairwise disjoint. 

For $C \in \mathcal C$, define $\mathcal P_C = \{X_{C,i}, Y_{C,i} : i \in \mathbb Z_3\}$.
For each clause $C$, fix an order to the variables\footnote{Because the boolean formula in $I$ is monotone, it does not have any negations, so we use the term ``variables'' instead of ``literals'' here. Some of the $z_{C,0}$ may of course be the same due to variables occurring in multiple clauses.} which occur in it as \(C = (z_{C,0}, z_{C,1}, z_{C,2}).\)
Set 
\begin{equation}\label{eq:construction}G_I = (\bigcup_{C \in \mathcal C} N(\mathcal P_C)) \cup (\bigcup_{(C,i) \in \mathcal C \times \mathbb{Z}_3}E(M_{z_{C,i}},X_{C,i})) \end{equation} where $N(\mathcal P_C)$ is the NAE gadget on $\mathcal P_C$ which involves $X_i$, $Y_i$, $x_i^b$, and $y_i^b$ in the definition of $N(\cdot)$ being instantiated by $X_{C,i}$, $Y_{C,i}$, $x^b_{C,i}$, and $y^b_{C,i}$, respectively.

Let $\chi_I $ be the coloring whose color classes are each of the pair classes thus constructed for $G_I$; let $\mathcal P_I$ denote the set of all such constructed pair classes.

\smallskip
\noindent
\textbf{Well-definedness.} We argue that the reduction creates a polynomial-size instance of \CE in which $\chi_I$ is a proper coloring of $G_I$ and every color class contains exactly two vertices. Let
\(
  n=|Z|\) and 
  \(m=|\mathcal C|\). Each gadget contains a constant number of vertices and edges, and $O(n+m)$ gadgets are added to $G_I$. 
Thus the construction is polynomial. Now every color class of $\chi_I$ is a pair class and therefore contains
exactly two vertices.  Moreover, none of the gadgets create edges whose endpoints both lie in the same pair class.
Hence, every prescribed color class is an independent set, and
$\chi_I$ is proper.

\begin{figure}[ht]
  \centering
  \begin{tikzpicture}[
      pairclass/.style={
        draw,
        dashed,
        rounded corners=2pt,
        minimum width=1.38cm,
        minimum height=7mm,
        inner sep=2pt,
        font=\small
      },
      master/.style={
        pairclass,
        fill=gray!12
      },
      occurrence/.style={
        pairclass,
        fill=blue!12
      },
      auxiliary/.style={
        pairclass,
        fill=orange!14
      },
      clausebox/.style={
        draw,
        rounded corners=4pt,
        line width=0.9pt,
        minimum width=5.6cm,
        minimum height=3.8cm
      },
      copywire/.style={
        line width=0.9pt
      },
      wirelabel/.style={
        circle,
        draw,
        fill=white,
        minimum size=5mm,
        inner sep=0pt,
        font=\scriptsize
      }, scale = 1, transform shape
    ]


    \node[font=\large\bfseries] at (6,7.65)
      {Example \textsc{Monotone NAE-3-SAT} instance};

    \node[font=\large] at (6,6.95)
      {$I=C_1\wedge C_2,\qquad
        C_1=(a,b,c),\qquad
        C_2=(c,d,e)$};


    \node[font=\large\bfseries] at (6,6.25)
      {Constructed $G_I$ and $\chi_I$};


    \node[master] (Ma) at (0.8,5.55) {$M_a$};
    \node[master] (Mb) at (2.4,5.55) {$M_b$};
    \node[master] (Mc) at (6.0,5.55) {$M_c$};
    \node[master] (Md) at (9.6,5.55) {$M_d$};
    \node[master] (Me) at (11.2,5.55) {$M_e$};


    \node[clausebox] (C1box) at (2.4,1.75) {};
    \node[clausebox] (C2box) at (9.6,1.75) {};


    \node[occurrence] (X10) at (0.8,3.20)
      {$X_{C_1,0}$};
    \node[occurrence] (X11) at (2.4,3.20)
      {$X_{C_1,1}$};
    \node[occurrence] (X12) at (4.0,3.20)
      {$X_{C_1,2}$};

    \node[occurrence] (X20) at (8.0,3.20)
      {$X_{C_2,0}$};
    \node[occurrence] (X21) at (9.6,3.20)
      {$X_{C_2,1}$};
    \node[occurrence] (X22) at (11.2,3.20)
      {$X_{C_2,2}$};


    \node[font=\small\bfseries] at (2.4,2.52)
      {The gadget $N(\mathcal P_{C_1})$};

    \node[auxiliary] (Y10) at (0.8,1.55)
      {$Y_{C_1,0}$};
    \node[auxiliary] (Y11) at (2.4,1.55)
      {$Y_{C_1,1}$};
    \node[auxiliary] (Y12) at (4.0,1.55)
      {$Y_{C_1,2}$};

    \node[
      font=\scriptsize,
      align=center,
      text width=4.8cm
    ] at (2.4,0.58)
      {This NAE gadget also contains edges for equality\\
       and containment gadgets};

    \node[font=\small\bfseries] at (9.6,2.52)
      {The gadget $N(\mathcal P_{C_2})$};

    \node[auxiliary] (Y20) at (8.0,1.55)
      {$Y_{C_2,0}$};
    \node[auxiliary] (Y21) at (9.6,1.55)
      {$Y_{C_2,1}$};
    \node[auxiliary] (Y22) at (11.2,1.55)
      {$Y_{C_2,2}$};

    \node[
      font=\scriptsize,
      align=center,
      text width=4.8cm
    ] at (9.6,0.58)
      {This NAE gadget also contains edges for equality\\
       and containment gadgets};


    \draw[copywire]
      (Ma.south) --
      node[wirelabel,pos=0.50] {$E$}
      (X10.north);

    \draw[copywire]
      (Mb.south) --
      node[wirelabel,pos=0.50] {$E$}
      (X11.north);

    \draw[copywire]
      (Mc.south west)
      to[out=-115,in=85]
      node[wirelabel,pos=0.52] {$E$}
      (X12.north);

    \draw[copywire]
      (Mc.south east)
      to[out=-65,in=95]
      node[wirelabel,pos=0.52] {$E$}
      (X20.north);

    \draw[copywire]
      (Md.south) --
      node[wirelabel,pos=0.50] {$E$}
      (X21.north);

    \draw[copywire]
      (Me.south) --
      node[wirelabel,pos=0.50] {$E$}
      (X22.north);

  \end{tikzpicture}
  \caption{
    The reduction applied to the instance
    $I=C_1\wedge C_2$, where
    $C_1=(a,b,c)$ and $C_2=(c,d,e)$.
    Each dashed rounded node denotes one pair class in the coloring $\chi_I$. Each large  box with a bold black outline denotes an NAE clause
    gadget; there is one such gadget for each clause of $I$. The lines labeled $E$ denote the edges involved in the equality gadgets
$E(M_{z_{C,i}},X_{C,i})$: each pair class for an occurrence of a variable $z$ in a clause is connected by an equality gadget edge to the pair class $M_z$ representing $z$. 
  }\label{fig:reduction}
\end{figure}
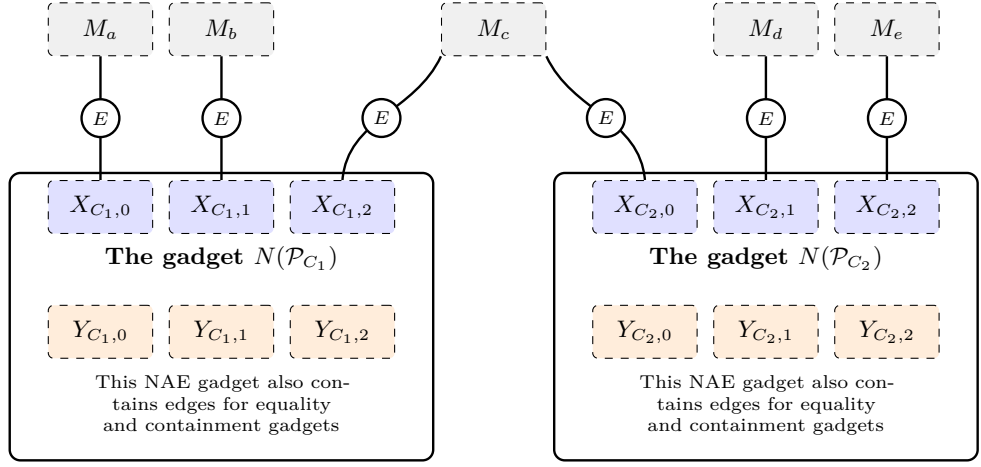

\begin{lemma}[Structure of the constructed instance]
\label{lem:construction-structure} 
For every clause $C \in \mathcal C$, define \(V_C = \bigcup_{P \in \mathcal P_C} P\). That is, $V_C$ consists of all of the vertices in the pair classes of $\mathcal P_C$. Define $\mathcal M = \bigcup_{z \in Z} \{M_z\}$ and  $V_\mathcal M = \bigcup_{z \in Z}M_z$. The following statements hold: \begin{enumerate}
\item For every $C \in \mathcal C$, $G_I[V_C] = N(\mathcal P_C)$. The restriction of $\chi_I$ to $V_C$ has 
color classes precisely the members of $\mathcal P_C$. Moreover, $G_I - V_\mathcal M = \bigsqcup_{C \in \mathcal C} N(\mathcal P_C)$.
\item For every $C \in \mathcal C$ and $i \in \mathbb Z_3$, $G_I[M_{z_{C,i}} \cup X_{C,i} ]= E(M_{z_{C,i}}, X_{C,i})$. Furthermore, every edge of $G_I$ which is incident to a vertex in $M_z$ has the form \(m_z^1x_{C,i}^0\) for some $C,i$ for which $z_{C,i} = z$.
\item Let $\tau : \mathcal P_I \to \{0,1\}$ be any orientation assignment, and, for each $C \in \mathcal C$, define $\tau_C = \tau|_{\mathcal P_C}$. Then $L_\tau(G_I) - \mathcal M = \bigsqcup_{C\in\mathcal C} L_{\tau_C}(N(\mathcal P_C))$ and $R_\tau(G_I) - \mathcal M = \bigsqcup_{C\in\mathcal C} R_{\tau_C}(N(\mathcal P_C))$.
\end{enumerate}

\end{lemma}
 Per Section~\ref{sec:union}, $G_I - V_\mathcal M$ is the subgraph of $G_I$ which is induced by deleting all vertices in $V_\mathcal M$ from $V(G_I)$.

 \begin{proof}[Proof sketch of Lemma~\ref{lem:construction-structure}]
Items~1 and~2 follow directly from the construction in \eqref{eq:construction}; they can also be seen by looking at  Figure~\ref{fig:reduction}. Item~3 follows because after deleting $\mathcal{M}$ from the digraphs, only pair classes belonging to  $\mathcal P_C$ for some $C \in \mathcal C$ remain in them, and there are no edges in $G_I$ between $V_C$ and $V_{C'}$ 
for unequal clauses $C$ and $C'$; therefore for $P \in \mathcal P_C$ and $Q \in \mathcal P_{C'}$, no vertex of one of these pair classes can distinguish the two vertices of another one of these pair classes, leaving no edges between these pair classes in the digraphs. Section~\ref{sec:lemma5full} of the Appendix contains a full check of items~1--3.
  \end{proof}

\subsection{Correctness}\label{sec:correct}
We show that $I$ is satisfiable if and only if $(G_I, \chi_I)$ is a \textsc{yes}-instance of \CE.

\smallskip
\noindent
\textbf{Forward Direction.}
Let $t: Z \to \{0,1\}$ be a satisfying assignment of $I$. Define $\tau:\mathcal P_I \to \{0,1\}$ by $\tau(M_z) = t(z)$ for $z \in Z$ and  $\tau(X_{C,i}) = \tau(Y_{C,i}) = t(z_{C,i})$ for every $C \in \mathcal C$ and $i \in \mathbb Z_3$. 
We first argue that every node $M_z$ is a source or a sink in each of
$L_\tau(G_I)$ and $R_\tau(G_I)$. Fix $C\in\mathcal C$ and
$i\in\mathbb Z_3$, and write $z=z_{C,i}$. By
Lemma~\ref{lem:construction-structure}, we have 
\(
G_I[M_z\cup X_{C,i}]=E(M_z,X_{C,i}).
\)
Since
\(
\tau(M_z)=\tau(X_{C,i})=t(z),
\)
Lemma~\ref{lem:basic-gadget-arcs} allows us to conclude that the only arc between these
two nodes, in each of $L_\tau(G_I)$ and $R_\tau(G_I)$, is
\(
M_z\to X_{C,i} \) when 
\(t(z)=1\),
and
\(
X_{C,i}\to M_z \) when \(t(z)=0.
\)
By item~2 of Lemma~\ref{lem:construction-structure}, every edge incident with a vertex of $M_z$ belongs to a gadget of the form $E(M_{z}, X_{C', i'})$ for $C' \in \mathcal C$ and $i' \in \mathbb Z_3$.
Therefore all arcs incident to $M_z$ point away from $M_z$ when
$t(z)=1$ and toward $M_z$ when $t(z)=0$. Thus $M_z$ is respectively a
source or a sink and cannot lie on a directed cycle.   Now let 
\(
  \mathcal M=\{M_z:z\in Z\}.
\)
For each clause $C$, let
\(
  \tau_C=\tau|_{\mathcal P_C}.
\)
By Lemma~\ref{lem:construction-structure}, we have 
\(
  L_\tau(G_I)-\mathcal M
  =
  \bigsqcup_{C\in\mathcal C}
  L_{\tau_C}(N(\mathcal P_C))
\)
and
\(
  R_\tau(G_I)-\mathcal M
  =
  \bigsqcup_{C\in\mathcal C}
  R_{\tau_C}(N(\mathcal P_C)).
\) 

Fix a clause
\(
  C=(z_{C,0},z_{C,1},z_{C,2})\).
For every $i\in\mathbb Z_3$,
\(
  \tau_C(X_{C,i})
  =
  \tau_C(Y_{C,i})
  =
  t(z_{C,i}).
\)
Because $t$ satisfies $I$, 
\(
  \tau_C(X_{C,0})\), \(
  \tau_C(X_{C,1})\), and \(
  \tau_C(X_{C,2})
\)
are not all equal. By Lemma~\ref{lem:naegadget}, both
\(
  L_{\tau_C}(N(\mathcal P_C)) \) and 
  \(R_{\tau_C}(N(\mathcal P_C))
\)
are acyclic. This holds for every clause.  Hence for both $L_\tau(G_I)$ and  $R_\tau(G_I)$, when we delete all nodes of the form $M_z$ for some $z \in Z$ from them, we are left with  disjoint unions
of acyclic digraphs, and thus the result after performing said deletions is acyclic.  Reintroducing the
$M_z$ cannot create a directed cycle, because each such node is
a source or a sink.  Thus $L_\tau(G_I)$ and $R_\tau(G_I)$ are both acyclic. Lemma~\ref{lem:respecttau} therefore gives a letter representation
of $G_I$ respecting $\chi_I$ whose linear order realizes $\tau$.
Hence, $(G_I,\chi_I)$ is a \textsc{yes}-instance of
\CE.

\smallskip\noindent\textbf{Backward Direction.}
Suppose that $(G_I,\chi_I)$ is a \textsc{yes}-instance of
\CE.  We can take the letter function of any $\chi_I$-respecting letter representation of $G_I$ to be $\chi_I$ itself, by the appropriate discussion in Section~\ref{subsec:color}. Accordingly, let
\(
  (\prec,\chi_I,D)
\)
be a letter representation of $G_I$. Define
\(
  \tau:\mathcal P_I\to\{0,1\}
\)
by
\(
  \tau(P)=1 \iff
  p^0\prec p^1
\)
for every pair class
\(
  P=\{p^0,p^1\}\in\mathcal P_I.
\)
Thus $\prec$ is consistent with the orientations prescribed by  $\tau$. By Lemma~\ref{lem:respecttau}, both $L_\tau(G_I)$ and $R_\tau(G_I)$ are acyclic. Define
\(
  t:Z\to\{0,1\}
\)
by
\(
  t(z)=\tau(M_z).
\) Fix $C\in\mathcal C$ and $i\in\mathbb Z_3$. By
Lemma~\ref{lem:construction-structure},
\(
G_I[M_{z_{C,i}}\cup X_{C,i}]
=
E(M_{z_{C,i}},X_{C,i}).
\)
If
$\tau(M_{z_{C,i}})\neq\tau(X_{C,i})$, then
Lemma~\ref{lem:basic-gadget-arcs} implies the existence of a directed $2$-cycle in one of
$L_\tau(G_I)$ and $R_\tau(G_I)$, contradicting their acyclicity.
Therefore,
\begin{equation}
\label{eq:occurrence-value}
\tau(X_{C,i})
=
\tau(M_{z_{C,i}})
=
t(z_{C,i})
\qquad
(C\in\mathcal C,\ i\in\mathbb Z_3).
\end{equation}
Now fix a clause $C\in\mathcal C$, and let
\(
\tau_C=\tau|_{\mathcal P_C}.
\)
By item~3 of Lemma~\ref{lem:construction-structure},
$L_{\tau_C}(N(\mathcal P_C))$ and
$R_{\tau_C}(N(\mathcal P_C))$ are induced subdigraphs of
$L_\tau(G_I)-\mathcal M$ and $R_\tau(G_I)-\mathcal M$, respectively.
They are thus acyclic, being subdigraphs of acyclic digraphs. Lemma~\ref{lem:naegadget} gives us that 
\(
\tau_C(X_{C,i})=\tau_C(Y_{C,i}) \) for all \(i\in\mathbb Z_3\), as well as  that
\(
\tau_C(X_{C,0}),
\tau_C(X_{C,1})\) and \(
\tau_C(X_{C,2})\)
are not all equal. By~\eqref{eq:occurrence-value}, 
\(
  t(z_{C,0}) \),
 \(t(z_{C,1}) \), and
  \(t(z_{C,2})\)
are not all equal.  Hence, $t$ satisfies $C$. Because $C$ was arbitrary, $t$ satisfies every clause of $I$.
Therefore $I$ is satisfiable.
\section{Reduction from Specialized Coloring Extension to Lettericity}\label{sec:colortoletter}

\begin{theorem}\label{thm:colortoletter}
Let $G$ be a graph, and let $\chi : V(G) \to [k]$ be a surjective proper coloring of $G$, with $k \ge 1$. There is a polynomial-time construction which produces a graph $H$ and integer $K = 2k$ such that $G$ has a letter representation respecting $\chi$ if and only if $\ell(H) \le K$.
\end{theorem}

\noindent
The rest of this section provides a proof of Theorem~\ref{thm:colortoletter}.

\subsection{Construction of the Graph $H$}\label{sec:constructH}
Set $K = 2k$ and $N = K + 1 = 2k+1$. For each $i \in [k]$ let $C_i = \chi^{-1}(i)$. Due to the surjectivity and properness of $\chi$, each $C_i$ is a non-empty independent set.

For every vertex $v \in V(G)$, the reduction creates an independent set $X_v$ of size $N$. These $X_v$ will come to have representatives in them that induce a subgraph in $H$ isomorphic to $G$ by the following rule: for $u \neq v$, if $uv \in E(G)$, create edges to make $X_u$ complete to $X_v$, and anticomplete otherwise.

For every $i\in[k]$, create two sets of vertices $M_i^0$ and $M_i^1$,
each of size $N$, and make
$M_i=M_i^0\cup M_i^1$ a clique of size $2N$. For distinct $i,j\in[k]$, make $M_i$ and $M_j$
anticomplete.

If $v\in C_i$, make $X_v$ complete to $M_i^1$ and anticomplete to
$M_i^0$. For every $j\in[k]\setminus\{i\}$, make $X_v$ complete to
$M_j^0$ and anticomplete to $M_j^1$. Thus $X_v$ is complete to exactly
one of the two halves of every $M_j$, and the index of the color class of $v$ -- which index is $i$ -- is the unique $i$ for which  $X_v$ is complete to the $1$-half of $M_i$ rather than the $0$-half thereof. 

The graph $H$ as constructed thus has $N(|V(G)| + 2k)$ vertices, and is thus of polynomial size because we have $k \le |V(G)|$ due to the non-emptiness of the $C_i$.

Figure~\ref{fig:second-reduction-example} provides an illustration of $H$ on an example source instance.

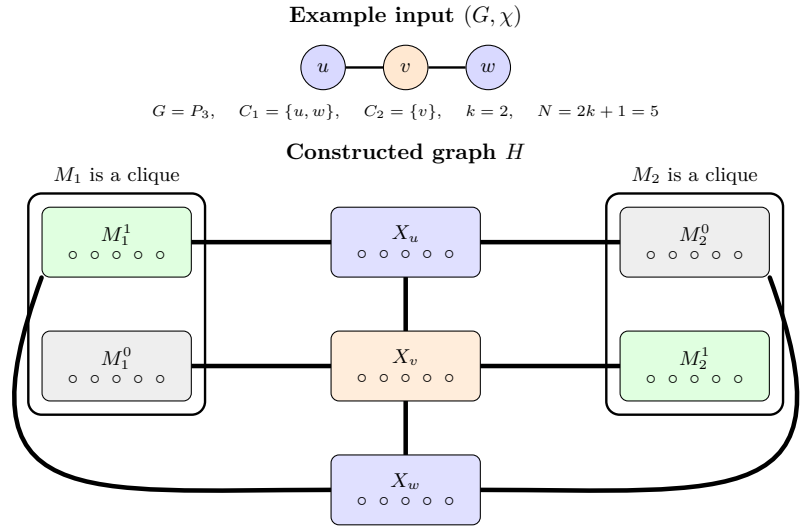
\begin{figure}[tb]
  \centering
  \begin{tikzpicture}[
      x=1cm,
      y=1cm,
      scale=0.84,
      transform shape,
      sourcevertex/.style={
        circle,
        draw,
        minimum size=7mm,
        inner sep=0pt,
        font=\small
      },
      classone/.style={sourcevertex,fill=blue!15},
      classtwo/.style={sourcevertex,fill=orange!18},
      setbox/.style={
        draw,
        rounded corners=3pt,
        minimum width=2.35cm,
        minimum height=1.10cm,
        align=center,
        inner sep=3pt,
        font=\small
      },
      cloudone/.style={setbox,fill=blue!12},
      cloudtwo/.style={setbox,fill=orange!14},
      markerzero/.style={setbox,fill=gray!13},
      markerone/.style={setbox,fill=green!12},
      markerblock/.style={
        draw,
        rounded corners=5pt,
        line width=0.9pt,
        inner sep=5pt
      },
      complete/.style={
        line width=1.7pt,
        line cap=round
      },
      sourceedge/.style={
        line width=0.9pt
      },
      note/.style={
        font=\scriptsize,
        align=center
      }
    ]


    \node[font=\bfseries] at (6,5.45)
      {Example input $(G,\mathcal \chi)$};

    \node[classone] (u) at (4.7,4.65) {$u$};
    \node[classtwo] (v) at (6.0,4.65) {$v$};
    \node[classone] (w) at (7.3,4.65) {$w$};

    \draw[sourceedge] (u)--(v)--(w);

    \node[note] at (6,3.95)
      {$G=P_3$, \quad
       $C_1=\{u,w\}$, \quad
       $C_2=\{v\}$, \quad
       $k=2$, \quad
       $N=2k+1=5$};

    \node[font=\bfseries] at (6,3.30)
      {Constructed graph $H$};

    %

    \node[markerone] (M11) at (1.45,1.90)
      {$M_1^1$\\[-1mm]
       $\circ\ \circ\ \circ\ \circ\ \circ$};

    \node[markerzero] (M10) at (1.45,-0.05)
      {$M_1^0$\\[-1mm]
       $\circ\ \circ\ \circ\ \circ\ \circ$};

    \node[markerzero] (M20) at (10.55,1.90)
      {$M_2^0$\\[-1mm]
       $\circ\ \circ\ \circ\ \circ\ \circ$};

    \node[markerone] (M21) at (10.55,-0.05)
      {$M_2^1$\\[-1mm]
       $\circ\ \circ\ \circ\ \circ\ \circ$};


    \begin{scope}[on background layer]
      \node[
        markerblock,
        fit=(M11)(M10),
        label={[font=\small]above:$M_1$ is a clique}
      ] (M1box) {};

      \node[
        markerblock,
        fit=(M20)(M21),
        label={[font=\small]above:$M_2$ is a clique}
      ] (M2box) {};
    \end{scope}


    \node[cloudone] (Xu) at (6,1.90)
      {$X_u$\\[-1mm]
       $\circ\ \circ\ \circ\ \circ\ \circ$};

    \node[cloudtwo] (Xv) at (6,-0.05)
      {$X_v$\\[-1mm]
       $\circ\ \circ\ \circ\ \circ\ \circ$};

    \node[cloudone] (Xw) at (6,-2.00)
      {$X_w$\\[-1mm]
       $\circ\ \circ\ \circ\ \circ\ \circ$};

    %

    \begin{scope}[on background layer]

      \draw[complete]
        (Xu.south)--(Xv.north);

      \draw[complete]
        (Xv.south)--(Xw.north);


      \draw[complete]
        (M11.east)--(Xu.west);

      \draw[complete]
        (M20.west)--(Xu.east);

      \draw[complete]
        (M11.south west)
        to[out=-110,in=180,looseness=1.8]
        (Xw.west);

      \draw[complete]
        (M20.south east)
        to[out=-70,in=0,looseness=1.8]
        (Xw.east);


      \draw[complete]
        (M10.east)--(Xv.west);

      \draw[complete]
        (M21.west)--(Xv.east);

    \end{scope}

  \end{tikzpicture}

  \caption{
    A schematic of the graph $H$ produced by the construction in Section~\ref{sec:constructH}  from
    $G=P_3$ with color classes
    $C_1=\{u,w\}$ and $C_2=\{v\}$.
    Each rounded rectangle represents a set of $N=5$ vertices.
    A thick segment between two sets denotes that the sets are complete to each other. 
    The absence of such  a segment denotes that they are anticomplete, except for that the
     halves of $M_i$ inside a common outer box together form the clique
    $M_i=M_i^0\cup M_i^1$.
    Each $X_z$ is an independent set. The three $X_z$ are connected as a path, in order to respect the connectivity of the source graph $P_3$. There are two $M_i$, one for each of the two color classes. Each of the $X_z$ is connected to exactly one of the halves of each $M_i$ -- and to which half of $M_i$  $X_z$ is complete depends on whether the color of $z$ is $i$. For example, the color class of $u$ is $C_1$, so $X_u$ is complete to the $1$-half of $M_1$, but the $0$-half of $M_2$. 
  }
  
  \label{fig:second-reduction-example}
\end{figure}

\subsection{Correctness Proof}
\label{sec:forward-colortoletter}

\noindent\textbf{Forward Direction.}
First, we prove that if $G$ has a letter representation respecting $\chi$, then $\ell(H) \le K$. Let $(\prec_G, \lambda, D_0)$ be such a representation of $G$. Because the representation respects $\chi$, for every $C_i$ each of its vertices is mapped to the same letter under $\lambda$, which we write as $\beta_i$. Because $C_i$ is an independent set in~$G$, if $|C_i| \ge 2$ then $(\beta_i, \beta_i) \notin D_0$. If $|C_i| = 1$, then we can still assume, w.l.o.g., that $(\beta_i, \beta_i) \notin D_0$, because no vertices outside of $C_i$ are mapped to $\beta_i$ under $\lambda$, so the decoder stipulates the exact same set of conditions as to which edges must be present in $G$ regardless of whether $(\beta_i, \beta_i)$ is in the decoder or not. Thus, \begin{equation}\label{eq:diagonalzero} (\beta_i, \beta_i) \notin D_0 \qquad (\text{for } i \in [k]).\end{equation}

We now construct a letter representation of $H$ whose alphabet consists
of the $2k$ pairwise distinct letters $\alpha_1,\ldots,\alpha_k,\beta_1,\ldots,\beta_k$.

Enumerate the vertices of $G$ such that $v_1\prec_G v_2\prec_G\cdots\prec_G v_n$. For the order on the vertices of $H$ which we are constructing to be part of a letter representation of $H$, construct the order to be such that in \begin{equation} \label{eq:orderm} M_1^0,M_2^0,\ldots,M_k^0,
X_{v_1},X_{v_2},\ldots,X_{v_n},
M_1^1,M_2^1,\ldots,M_k^1,\end{equation} each of the vertices in one set occurs before all of the vertices in all sets enumerated after it. It is immaterial how the ordering orders the vertices within each of the sets themselves. 

For the letter function, map all $2N$ vertices of
$M_i=M_i^0\cup M_i^1$ to $\alpha_i$, and map all $N$ vertices of $X_v$
to $\beta_i$ whenever $v\in C_i$.

For $i,j\in[k]$, define the decoder $\widehat D$ according to the
following rules:
\begin{align}
  (\beta_i,\beta_j)\in\widehat D
  &\quad\Longleftrightarrow\quad
  (\beta_i,\beta_j)\in D_0,
  \label{eq:decoder-data-data}\\
  (\alpha_i,\alpha_j)\in\widehat D
  &\quad\Longleftrightarrow\quad
  i=j,
  \label{eq:decoder-marker-marker}\\
  (\alpha_j,\beta_i)\in\widehat D
  &\quad\Longleftrightarrow\quad
  j\ne i,
  \label{eq:decoder-left-marker}\\
  (\beta_i,\alpha_j)\in\widehat D
  &\quad\Longleftrightarrow\quad
  j=i.
  \label{eq:decoder-right-marker}
\end{align}

These rules cover every pair of letters in the alphabet.

We now verify that the letter representation constructed in this subsection is a valid letter representation of $H$. We do a case analysis on what any pair of distinct vertices in $V(H)$ could be, in each case showing that $\widehat D$ says that there should be an edge between them exactly when there is such an edge in $E(H)$.

\begin{enumerate}
\item For two vertices in $X_v$, they are assigned $\beta_i$ for $i$ such that $v \in C_i$. By \eqref{eq:diagonalzero}, $(\beta_i, \beta_i) \notin D_0$, and thus by \eqref{eq:decoder-data-data}, $(\beta_i, \beta_i) \notin \widehat D$, according with how $X_v$ is an independent set in $H$. 
\item Consider one vertex in $X_{v_p}$ and one vertex in $X_{v_q}$. Assume 
$p<q$ (which is equivalent to $v_p \prec_G v_q$), so \eqref{eq:orderm} gives that the vertex in $X_{v_p}$ precedes the one in $X_{v_q}$ in the ordering for the constructed letter representation. By construction, there is an edge between these two vertices if and only if $v_p$ and $v_q$ are adjacent in $G$, and $v_p$ and $v_q$ are adjacent in $G$ if and only if $(\beta_i, \beta_j) \in D_0$ for $i,j$ such that $v_p \in C_i$ and $v_q \in C_j$. This last condition  is equivalent to $(\beta_i, \beta_j) \in \widehat D$. $\beta_i, \beta_j$ are exactly the two letters under the constructed letter function of the vertices we are considering here, and thus we have that they are adjacent if and only if the tuple containing the letter for the lesser one in the new order followed by the letter for the greater one is in $\widehat D$.
\item
A pair of vertices in $M_i$ are necessarily adjacent in $H$ because
$M_i$ is a clique. They are represented by $\alpha_i$, and
$(\alpha_i,\alpha_i)\in\widehat D$ by
\eqref{eq:decoder-marker-marker}.

\item
Two vertices, one of which is in $M_i$ and the other of which is in
$M_j$ for $i\ne j$, are necessarily nonadjacent, in agreement with
\eqref{eq:decoder-marker-marker}.

\item
For a vertex in $M_j^0$ and a vertex in $X_v$, where $v\in C_i$:
every vertex of $M_j^0$ precedes every vertex of $X_v$. By construction,
$M_j^0$ is complete to $X_v$ exactly when $j\ne i$, which is exactly
the condition in \eqref{eq:decoder-left-marker}.

\item
For a vertex in $X_v$ and a vertex in $M_j^1$, where $v\in C_i$:
every vertex of $X_v$ precedes every vertex of $M_j^1$. By construction,
$X_v$ is complete to $M_j^1$ exactly when $j=i$, which is exactly the
condition in \eqref{eq:decoder-right-marker}.
\end{enumerate}
\smallskip
\noindent\textbf{Backward Direction.}
Assume that $H$ has a letter representation $(\prec, \lambda: V(H) \to \Sigma, D)$ such that $|\Sigma| \le K$. Each set $X_v$ or $M_i^b$ consists of $N=K+1$ elements, so by the pigeonhole principle, for each of these sets there exists a letter in $\Sigma$ which has two distinct preimages in it. For each $v \in V(G)$, select some $x_v \in X_v$ such that $\lambda(x_v)$ has two distinct preimages in $X_v$. For each $i \in [k]$, select some $p_i \in M_i^0$ and $q_i \in M_i^1$ such that $\lambda(p_i)$ has two distinct preimages in $M_i^0$ and $\lambda(q_i)$ has two distinct preimages in $M_i^1$.

Let $t = |\{ \lambda(x_v) : v \in V(G) \}|$. Because $X_v$ is an independent set, we have $(\lambda(x_v), \lambda(x_v)) \notin D$; whereas because $M_i$ is a clique, $(\lambda(p_i), \lambda(p_i)), (\lambda(q_i), \lambda(q_i)) \in D$; therefore for all $v \in V(G)$, for all $i \in [k] $ $\lambda(x_v) \neq \lambda(p_i)$ and $\lambda(x_v) \neq \lambda(q_i)$. As a result, the $p_i$ and $q_i$ must be mapped (under $\lambda$) to at most $|\Sigma| - t \le 2k - t$ letters. Now  each of the $\lambda(p_i)$ must be distinct because $M_i$ and $M_j$ are anticomplete for $i \neq j$; their anticompleteness means that if $\lambda(p_i)$ has a preimage in both $M_i$ and $M_j$, $(\lambda(p_i),\lambda(p_i)) \notin D$, whereas we already mentioned how $(\lambda(p_i),\lambda(p_i)) \in D$. The same reasoning in this anticompleteness argument establishes that more generally, for $ i\neq j$ $\{\lambda(p_i), \lambda(q_i)\} \cap \{\lambda(p_j), \lambda(q_j)\} = \emptyset$. Given this and that there are $2k$ of the $p_i$ and $q_i$, it must be that $|\{ i \in [k] \mid \lambda(p_i) = \lambda(q_i)\}| \ge t$ so that the $p_i$ and $q_i$ can be mapped to at most $2k- t$ letters. Let $S = \{ i \in [k] \mid \lambda(p_i) = \lambda(q_i)\}$. 

\begin{lemma}\label{lem:checkerboard}
    
For each $i \in S$, if $v\in C_i$, $w\in C_j$, and $i\ne j$, then $\lambda(x_v) \neq \lambda(x_w)$. 
\end{lemma}

\begin{proof}

Let $i \in S$ be arbitrary and suppose toward a contradiction that $\lambda(x_v)=\lambda(x_w)$, where
$v\in C_i$, $w\in C_j$, and $j\ne i$.  By construction, $x_vq_i, x_wp_i \in E(H)$ whereas $x_vp_i, x_wq_i \notin E(H)$, so $p_i, q_i \in \Delta_H(x_v,x_w)$ and  $x_v, x_w \in \Delta_H(p_i,q_i)$. Thus the forward direction of Lemma~\ref{lem:orderofdistinguishers} implies that $\{p_i,q_i\} \subseteq I_\prec(x_v,x_w)$ and $\{x_v,x_w\} \subseteq I_\prec(p_i,q_i)$, which is a contradiction. 
\end{proof}

Now let $J \subseteq S$ with $|J| = t$, which is possible due to how we established $|S| \ge t$. For each $i \in J$, select some $v_i \in C_i$ (which is possible by the surjectivity of $\chi$) and define $\beta_i = \lambda(x_{v_i})$. The $\beta_i$ are pairwise distinct by Lemma~\ref{lem:checkerboard}. Now $\{\beta_i:i\in J\} \subseteq \{\lambda(x_v):v\in V(G)\} $, but also these two sets have the same cardinality of $t$ so we have $\{\lambda(x_v):v\in V(G)\} = \{\beta_i:i\in J\}$. Consequently, for any $v \in V(G)$, $\lambda(x_v)$ occurs as one of the $\beta_i$; and it occurs as exactly one of them because of the aforementioned pairwise distinctness of the $\beta_i$. 

Now if there exists $j \in [k] \setminus J$, then because $C_j$ is nonempty there exists $w \in C_j$; by the previous paragraph there is $i \in J$ for which $\lambda(x_w) = \beta_i = \lambda(x_{v_i})$, which contradicts Lemma~\ref{lem:checkerboard} because the fact that $i \in J$ but $j \notin J$ implies that the color classes of  $v_i$ and $w$ are different, and thus that $x_w \neq x_{v_i}$. Therefore $J = [k]$ and $t = k$.

By the construction  of $H$, the vertices $x_v$ induce a copy of $G$ in $H$. Restricting the letter representation $(\prec, \lambda, D)$ of $H$ to be with respect to the elements of $\{ x_v : v \in V(G) \}$ only and relabeling $x_v$ as $v$ in the letter representation gives a $\chi$-respecting letter representation of~$G$: the $\beta_i$ are pairwise distinct and every vertex of $V(G)$ in color class~$C_i$ shares the letter~$\beta_i$ under this letter representation: if $w \in C_i$, then $\lambda(x_w) = \beta_j$ for $j \in [k]$; if $j \neq i$, then
$\lambda(x_w)=\beta_j=\lambda(x_{v_j})$,
contradicting Lemma~\ref{lem:checkerboard} as
$j\in J\subseteq S$ and $w,v_j$ belong to different
color classes.

\section{Complexity Results for \textsc{Lettericity} and Other Problems}
\subsection{\textsc{Lettericity Complexity Results}}
Taken together, Sections~\ref{sec:sattocolor} and~\ref{sec:colortoletter} imply that \textsc{Lettericity} is \textsf{NP}-hard. They also help show that \textsc{Lettericity} is \textsf{NP}-complete.

\begin{theorem}\label{thm:np-complete}
\textsc{Lettericity} is \textsf{NP}-complete.
\end{theorem}

\begin{proof}[Proof sketch]
For a  \textsc{yes}-instance $(G,k)$ of \textsc{Lettericity}, a certificate is a letter representation $(\prec, \lambda: V(G) \to \Sigma, D)$ of $G$ for which  $|\Sigma| \le k$. It can be checked in polynomial time whether the certificate is a valid letter representation of~$G$ which satisfies $|\Sigma| \le k$; in particular, verifying that for all $u \prec v$, \(uv \in E(G) \iff (\lambda(u), \lambda(v)) \in D\) can be done by only iterating through $n \choose 2$ items. Section~\ref{sec:full-np-complete} gives a more detailed proof which contains more discussion as to how the encoding of the certificate leads to polynomial-time computation of the relevant verification processes. \end{proof}

\begin{theorem}[ETH Lower Bound for \textsc{Lettericity}\label{thm:eth-lowerbound}] 
Unless the Exponential Time Hypothesis is incorrect, there cannot exist a deterministic algorithm which, given a graph $G$ on $n$ vertices and a positive integer $k$,
decides whether $\ell(G) \le k$  in time $2^{o(n)}$, even when $n = 6k$.
\end{theorem}
\begin{proof}[Proof sketch] We give just a proof sketch here; a full proof for Theorem~\ref{thm:eth-lowerbound} can be found in Section~\ref{sec:full-cor-eth-lowerbound} of the Appendix. Assume the ETH. Under the ETH, there cannot exist a deterministic algorithm which decides \textsc{Monotone-NAE3-SAT} in $2^{o(m)}$ time, where $m$ is the number of clauses in its input instances of \textsc{Monotone-NAE3-SAT}~\cite[Lemma~5.4]{KnoMas2018}.

Let $R$ be the reduction from  \textsc{Monotone-NAE3-SAT} to \textsc{Lettericity} formed by composing the reductions in Sections~\ref{sec:sattocolor} and~\ref{sec:colortoletter}. $R$ can have a simple  postprocessing step added to it such that when it takes as input  a \textsc{Monotone-NAE3-SAT} instance $I$ with $n'$ variables and $m$ clauses it outputs $(H,K)$ for $|V(H)| = 6K$ and it is still a valid Karp-reduction from \textsc{Monotone-NAE3-SAT} to \textsc{Lettericity}. This postprocessing step does not change the parameter~$K$, 
which by construction satisfies $K = \Theta(n' + m)$. We can assume $n' \le 3m$ because unused variables in $I$ can be deleted, giving $K = \Theta(m)$ and thus $|V(H)| = \Theta(m)$. If Theorem~\ref{thm:eth-lowerbound} was false, $\ell(H) \le K$ can be decided in $2^{o(|V(H)|)}$ time, which gives a $2^{o(m)}$ deterministic algorithm to decide \textsc{Monotone-NAE3-SAT} --- a contradiction. \end{proof}

\newcommand{\duall}{\overline{k}}
Even though not completely tight with the lower bound, we can show that \LET can be solved in $n^{\Oh(n)}$~time.
In fact, the exponent will only depend on the \emph{dual parameter}~$\duall \coloneqq n-k$.

\begin{theorem}\label{thm:fpt-dual}
\LET can be solved in $n^{\Oh(\duall)}$~time and in $2^{\Oh(\duall^2)} + \Oh(n)$~time. 
\end{theorem}

Note that this implies that~\LET admits an \textsf{FPT}-algorithm both for its most natural parameter~$k$ (see~\cite{AleKLZ2026}) and the dual parameter~$\duall$, a rarely observed phenomenon for \textsf{NP}-hard parameterized problems on general graphs.
The \textsf{FPT}-result itself already immediately follows from a combinatorial upper bound on lettericity obtained by~\cite{ManVat2024}. The main point of our result is the improved running time.
The rough idea for the algorithm is to enumerate (up to isomorphism) all colorings~$\chi$ that use exactly~$k$ colors, together with 
all linear orders of the set~$W$ of vertices for which the color class under~$\chi$ has size at least~$2$.
For each such choice of~$\chi$ and~$\prec_W$, we then decide whether there is a linear order~$\prec$ of~$V$ that extends~$\prec_W$ together with a decoder, such that~$(\prec,\chi,D)$ is a letter representation of~$G$.
As we can show, this can be done in polynomial time.
We then achieve the stated algorithms by showing that the number of choices for~$\chi$ and~$\prec_W$ is~$n^{\Oh(\duall)}$ (and~$2^{\Oh(\duall^2)}$ after some linear-time preprocessing step). For more details, see \Cref{sec proof fpt} of the Appendix.



\subsection{Complexity Results for Other Problems}
Along the way toward proving \textsf{NP}-hardness of \textsc{Lettericity}, we also showed the next result.\begin{theorem}\label{thm:coloring-extension-npcomplete}
\CE is \textsf{NP}-complete even when the coloring is proper and its color classes are all of size exactly two.
    \end{theorem}

    \begin{proof}
    The reduction in Section~\ref{sec:sattocolor} establishes \textsf{NP}-hardness for the appropriate \CE variant. Given an instance $(G, \chi: V(G) \to [k])$ of \CE, it can be verified that a certificate $(\prec, \lambda, D)$ witnesses that $(G, \chi)$ is a \textsc{yes}-instance of \CE by checking first whether $\lambda$ is constant on each color class of $\chi$ and assigns different letters to different color classes and then using an \textsf{NP}-verifier for \textsc{Lettericity} (the existence of which is given by Theorem~\ref{thm:np-complete}) to check whether $(\prec, \chi, D_\chi)$ is a  valid certificate for the instance $(G,k)$ of \textsc{Lettericity}, where $D_\chi$ is the result of relabeling the entries of~$D$ to accord with $\chi$ (
    i.e., if there exist pairwise distinct letters $\beta_1, \ldots, \beta_k \in \Sigma$ 
    with $\lambda(v) = \beta_{\chi(v)}$ for all $v\in V(G)$, then $D_\chi$ is the result of relabeling the entries of $D$ according to $\beta_i \mapsto i$).  \end{proof}

The same paper which raised as an open question the complexity of \CE~(\cite{GroMorSac2026}) also asked what the complexity of \WE was.  \WE asks, for a given $n$-vertex graph $G$ and given word $(w_1, \ldots, w_n) \in \Sigma^n$, whether $G$ has a letter representation $(\prec, \lambda: V(G) \to \Sigma, D)$ for which after enumerating $V(G)$ as $v_1 \prec v_2 \prec \cdots \prec v_n$, for all $i \in [n]$, $\lambda(v_i) = w_i$. The theorem below is proved in Section~\ref{sec:word-extension-np-complete} of the Appendix. 
\begin{theorem}\label{thm:word-extension-np-complete}
\WE is \textsf{NP}-complete even when $G$ is a disjoint union of equally sized cliques.
\end{theorem}

Additionally, \cite{GroMorSac2026} introduced the analysis of extension problems with \DE, which they showed to be \textsf{GI}-hard. \DE asks, for a given $n$-vertex graph $G$, alphabet $\Sigma$, and decoder~$D \subseteq \Sigma^2$, whether $G$ has a letter representation $(\prec, \lambda: V(G) \to \Sigma, D')$ with~$D = D'$. 
We strengthen this known hardness by proving that~\DE is even \textsf{NP}-complete by presenting a reduction from~\CE in which the decoder enforces the prescribed color classes of the input coloring to the~\CE-instance.
The theorem below is proved in Section~\ref{sec DE hard} of the Appendix. 
\begin{theorem}\label{thm:decoder-extension-np-complete}
\DE is \textsf{NP}-complete.
\end{theorem}

Together with results from \cite{GroMorSac2026},  this thus completes the complexity picture for all combinations of prescribed parts of the solution. An illustration can be found in \Cref{sec:extra-figures}. 

\section{Conclusions}

In this paper, it is proven that the lettericity problem is \textsf{NP}-complete. Notice that this is the best that we can achieve: as it is known that the question to determine if the lettericity of a given graph is at most~$k$ is in $\textsf{FPT}$ when parameterized by~$k$, we cannot hope for an $\textsf{NP}$-hardness result for any fixed~$k$ (in contrast to the well-known \textsc{Colorability} problem) unless $\textsf{P} = \textsf{NP}$; see \cite{AleKLZ2026}. Yet, it would be interesting to better understand the structure of graphs of lettericity at most~$k$ for fixed~$k$. For instance, as argued in \cite{AleFKLVZ2022}, threshold graphs have lettericity at most~$2$.

Also, it would be interesting to solve the lettericity problem on special graph classes. Recall the discussion of Vatter's proof~\cite{Vat2026} above. It is at least implicitly asked if this is possible for classes as small as cographs in~\cite{AleLozWer2022}. By the mentioned parameterized complexity result, it also follows that graph classes with a potentially \textsf{NP}-hard lettericity problem must contain graphs of arbitrarily large lettericity. In this context, the studies of~\cite{FerVat2022} give an important characterization of these potentially interesting graph classes.
Several graph classes have been already studied concerning lettericity; we point to \cite{AleALM2023,AleLozWer2022,AleLWZ2020}.
Links to classes of permutations also offer interesting connections, see \cite{AleLWZ2020,Bra2025,BriVat2022}.
Similarly, complexity studies on graph classes for the related parameter \emph{thinness}, only recently proved to be an \textsf{NP}-hard graph parameter~\cite{Shi2025}, are also lacking on many natural graph classes.

The lower bound of Theorem~\ref{thm:eth-lowerbound} suggests a natural question. In~\cite[Theorem~16]{AleKLZ2026}, a $2^{\mathcal{O}(k^2 2^{2k})}n^3$ algorithm for \textsc{Lettericity} is given. Can this be improved by showing some 
$2^{\mathcal{O}(k)}n^{O(1)}$ algorithm for \textsc{Lettericity}? If so, there would also be a $2^{\mathcal{O}(n)}$ algorithm, as \textsc{Lettericity} instances with $k \ge n$ are trivially \textsc{yes}-instances. This would, up to constant factors in the exponent, match the lower bound of $2^{o(n)}$ under ETH which Theorem~\ref{thm:eth-lowerbound} gives, 
getting rid of the logarithm in the naive $2^{\mathcal{O}(n \log n)}$ algorithm for \textsc{Lettericity}.

Finally, it would be interesting to determine the complexity status of the decision problems associated to related parameters, like generalized lettericity \cite{FenFMRS2025}, uniformicity and distinguishing number~\cite{Loz2023}.




\bibliography{ab,local}

@string{alg = {Algorithmica}}

@string{dam = "Discrete Applied Mathematics"}

@string{dm = "Discrete Mathematics"}

@string{jgth = "Journal of Graph Theory"}

@string{pjm = "Pacific Journal of Mathematics"}

@string{siamdm = {{SIAM} Journal of Discrete Mathematics}}

@string{tcs = "Theoretical Computer Science"}

@article{AleALM2023,
  author       = {B. Alecu and
                  A. Atminas and
                  V. V. Lozin and
                  D. S. Malyshev},
  title        = {Combinatorics and Algorithms for Quasi-Chain Graphs},
  journal      = alg,
  volume       = {85},
  number       = {3},
  pages        = {642--664},
  year         = {2023},
  doi          = {10.1007/S00453-022-01019-6},
}

@article{AleFKLVZ2022,
  author       = {B. Alecu and
                  R. Ferguson and
                  M. M. Kant{\'{e}} and
                  V. V. Lozin and
                  V. Vatter and
                  V. Zamaraev},
  title        = {Letter Graphs and Geometric Grid Classes of Permutations},
  journal      = siamdm,
  volume       = {36},
  number       = {4},
  pages        = {2774--2797},
  year         = {2022},
  doi          = {10.1137/21M1449646},
}

@article{AleKLZ2026,
  author       = {B. Alecu and
                  M.~M. Kant{\'{e}} and
                  V. V. Lozin and
                  V. Zamaraev},
  title        = {Lettericity of graphs: an {FPT} algorithm and a bound on the size
                  of obstructions},
  journal      = alg,
  volume       = {88},
  number       = {1},
  pages        = {2:1--2:27},
  year         = {2026},
  doi          = {10.1007/S00453-025-01341-9},
}

@article{AleLozMal2024,
author = {Alecu, B. and Lozin, V. and Malyshev, D.},
title = {Critical properties of bipartite permutation graphs},
journal = jgth,
volume = {105},
number = {1},
pages = {34--60},
doi = {https://doi.org/10.1002/jgt.23011},
year = {2024}
}

@article{AleLozWer2022,
  author       = {B. Alecu and
                  V.~V. Lozin and
                  D. de Werra},
  title        = {The micro-world of cographs},
  journal      = dam,
  volume       = {312},
  pages        = {3--14},
  year         = {2022},
  doi          = {10.1016/J.DAM.2021.11.004},
}

@article{AleLWZ2020,
  author       = {B. Alecu and
                  V.~V. Lozin and
                  D. de Werra and
                  V. Zamaraev},
  title        = {Letter graphs and geometric grid classes of permutations: Characterization
                  and recognition},
  journal      = dam,
  volume       = {283},
  pages        = {482--494},
  year         = {2020},
  doi          = {10.1016/J.DAM.2020.01.038},
}

@inproceedings{AleVilLoz2026,
  author       = {B. Alecu and
                  P.~B. Villafana and
                  V.~V. Lozin},
  editor       = {J. Goedgebeur and
                  P. Rzazewski},
  title        = {Cycles in Unions of Transitive Tournaments},
  booktitle    = {52nd International Workshop on Graph-Theoretic Concepts in Computer
                  Science, {WG}},
  series       = {LIPIcs},
  volume       = {376},
  pages        = {3:1--3:16},
  publisher    = {Schloss Dagstuhl - Leibniz-Zentrum f{\"{u}}r Informatik},
  year         = {2026},
  doi          = {10.4230/LIPICS.WG.2026.3},
}

@article{AtmLoz2024,
  author       = {A. Atminas and
                  V.~V. Lozin},
  title        = {Deciding atomicity of subword-closed languages},
  journal      = tcs,
  volume       = {1003},
  pages        = {114595:1--11},
  year         = {2024},
  doi          = {10.1016/J.TCS.2024.114595},
}

@article{BodZan2020,
  author       = {H.~L. Bodlaender and
                  T.~C. van der Zanden},
  title        = {On the exact complexity of polyomino packing},
  journal      = tcs,
  volume       = {839},
  pages        = {13--20},
  year         = {2020},
  doi          = {10.1016/J.TCS.2020.05.025},
}

@Article{Bra2025,
title   ="Decidability in Geometric Grid Classes of Permutations",
author  ="S. Braunfeld",
pages   ="987--1000",
journal ={Proceedings of the American Mathematical Society},
year    ="2025",
volume  ="153",
number  ="3",
doi = {10.1090/proc/17083}
}

@article{BriVat2022,
  author       = {R. Brignall and
                  V. Vatter},
  title        = {Labelled well-quasi-order for permutation classes},
  journal      = {Combinatorial Theory},
  volume       = {2},
  number       = {3},
pages = {1--55},
  year         = {2022},
  doi          = {10.5070/C62359178},
}

@article{DarDoc2020,
  author    = {A. Darmann and
               J. D{\"{o}}cker},
  title     = {On a simple hard variant of Not-All-Equal 3-{SAT}},
  journal   = tcs,
  volume    = {815},
  pages     = {147--152}, year = {2020}, doi={10.1016/j.tcs.2020.02.010}
}

@techreport{FenFFKS2026,
  author       = {Z. Feng and H. Fernau and
                  P. Fleischmann and
                  P. Kindermann and
                  S.~C. Sacher},
      title={Determining Factorial Speed Fast}, 
      year={2026},
month = feb,
  institution   = {ArXiv,  Cornell University},
  number    = {2602.24064},
      yeprint={2602.24064},
      yarchivePrefix={arXiv},
      yprimaryClass={cs.DM},
      yurl={https://arxiv.org/abs/2602.24064}, 
}

@inproceedings{FenFMRS2025,
  author       = {Z. Feng and
                  H. Fernau and
                  K. Mann and
                  I. Raman and
                  S.~C. Sacher},
  editor       = {D.~R. Gaur and
                  R. Mathew},
  title        = {Generalized Lettericity of Graphs},
  booktitle    = {Algorithms and Discrete Applied Mathematics - 11th International Conference,
                  {CALDAM}},
  series       = {LNCS},
  volume       = {15536},
  pages        = {134--146},
  publisher    = {Springer},
  year         = {2025},
  doi          = {10.1007/978-3-031-83438-7_12},}

@article{Fer2020,
  author       = {R. Ferguson},
  title        = {On the lettericity of paths},
  journal      = {Australasian Journal of Combinatorics},
  volume       = {78},
  pages        = {348--351},
  year         = {2020},
  url          = {http://ajc.maths.uq.edu.au/pdf/78/ajc\_v78\_p348.pdf},
}

@article{FerVat2022,
  author       = {R. Ferguson and
                  V. Vatter},
  title        = {Letter graphs and modular decomposition},
  journal      = dam,
  volume       = {309},
  pages        = {215--220},
  year         = {2022},
  doi          = {10.1016/J.DAM.2021.11.007},
}

@inproceedings{GroMorSac2026,
  author       = {M. Grobler and N. Morawietz and S.~C. Sacher},
  editor       = {V. Brattka and H. Fernau and L. Galeotti},
  title        = {Towards Settling the Complexity of the Lettericity Problem},
  booktitle    = {Timeless Machines: Computability Across Eras; 22nd Conference on Computability in Europe, CiE},
  pages        = {304--319},
  publisher    = {Springer},
volume={16674},
  series = {LNCS},
  year         = {2026},
doi ={10.1007/978-3-032-31348-5_20}
}

@article{JanLanLan2016,
  author    = {K. Jansen and
               F. Land and
               K. Land},
  title     = {Bounding the Running Time of Algorithms for Scheduling and Packing
               Problems},
  journal   = siamdm,
  volume    = {30},
  number    = {1},
  pages     = {343--366},
  year      = {2016},
  doi          = {10.1137/140952636},
}

@article{KnoMas2018,
title = {Computational complexity of distance edge labeling},
journal = dam,
volume = {246},
pages = {80--98},
year = {2018},
doi = {https://doi.org/10.1016/j.dam.2017.01.007},
author = {D. Knop and T. Masařík},
}

@article{Loz2023,
  author    = {V.~V. Lozin},
  title     = {Hereditary classes of graphs: {A} parametric approach},
  journal   = dam,
  volume    = {325},
  pages     = {134--151},
  year      = {2023},
doi={10.1016/j.dam.2022.10.016}
}

@article{ManVat2024,
      title={Bounds on the lettericity of graphs}, 
      author={S. Mandrick and V. Vatter},
  journal      = {The Electronic Journal of Combinatorics},
  volume       = {31},
  number       = {4},
pages={53:1--53:8},
  year         = {2024},
  doi          = {10.37236/12411},
}

@article{Pet2002,
title = {Letter graphs and well-quasi-order by induced subgraphs},
journal = dm,
volume = {244},
number = {1},
pages = {375--388},
year = {2002},
doi = {10.1016/S0012-365X(01)00094-2},
author = {M. Petkovšek},}

@article{Shi2025,
	author = {Y. Shitov},
  title     = {Graph thinness: a lower bound and complexity},
  journal   = pjm,
  volume    = {339},
  number    = {2},
doi={10.2140/pjm.2025.339.333},
  pages     = {333--343},
  year      = {2025},
}

@TECHREPORT{TeyRap2008,
	AUTHOR = "N. Teypaz and C. Rapine",
	TITLE = "Graph decomposition into paths under length constraints",
	INSTITUTION = "HAL Open Science",
	YEAR = 2008,
	month = may,
	number = "hal00278189",
	url = "https://hal.science/hal-00278189v1",
	copy = ""
}

@techreport{Vat2026,
  author       = {V. Vatter},
      title={Deciding lettericity is {NP}-complete, even for four-colorable comparability graphs}, 
institution = {ArXiv},
      year={2026},
      number={2610.06482},
doi = {10.48550/ARXIV.2610.06482},
      eprint={2610.06482},
      archivePrefix={arXiv},
      primaryClass={cs.CO},
      yurl={https://arxiv.org/abs/2610.06482}, 
}

@BOOK{GarJoh79,
	AUTHOR = "M. R. Garey and D. S. Johnson",
	TITLE = "Computers and Intractability",
	PUBLISHER = "New York: Freeman",
	YEAR = 1979,
	month = "",
	volume = "",
	edition = "",
	series = "",
	doi = "10.5555/578533"
}
\clearpage
\appendix
\section{Deferred Proofs}

\subsection{Proof of Lemma~\ref{lem:basic-gadget-arcs}}\label{sec:arc-table-proof}
For reference, the statement of Lemma~\ref{lem:basic-gadget-arcs} is reiterated below. 

Let $A = \{a^0, a^1\}$ and $B = \{b^0, b^1\}$  be pair classes. Let $\tau:\{A,B\}\to\{0,1\}$ be an orientation assignment. The arc sets of $L_\tau(E(A,B))$ and $R_\tau(E(A,B))$ are given by the following table:
\[
\begin{array}{c|c|c}
(\tau(A),\tau(B))
  & L_\tau(E(A,B))
  & R_\tau(E(A,B)) \\
\hline
(1,1) & \{A\to B\} & \{A\to B\} \\
(0,0) & \{B\to A\} & \{B\to A\} \\
(0,1) & \{A\to B,\ B\to A\} & \varnothing \\
(1,0) & \varnothing & \{A\to B,\ B\to A\}
\end{array}
\]
Moreover, for the containment gadget, the only arc of $L_\tau(C(A,B))$ is $A\to B$, and the only arc of $R_\tau(C(A,B))$ is $B\to A$.

\begin{proof}[Proof of Lemma~\ref{lem:basic-gadget-arcs}]
We have 
\(
\Delta_{E(A,B)}(a^0,a^1)=\{b^0\} \) and \(\Delta_{E(A,B)}(b^0,b^1)=\{a^1\}\) in the equality gadget.
Consequently, directly from the definitions of $L_\tau$ and~$R_\tau$,
\[
\begin{aligned}
A\to B\in L_\tau(E(A,B))&\iff \tau(B)=1,\\
B\to A\in L_\tau(E(A,B))&\iff \tau(A)=0,\\
A\to B\in R_\tau(E(A,B))&\iff \tau(A)=1,\\
B\to A\in R_\tau(E(A,B))&\iff \tau(B)=0.
\end{aligned}
\]
These four equivalences give the table for the arc sets of $L_\tau(E(A,B))$ and $R_\tau(E(A,B))$.

For the containment gadget,
\(
\Delta_{C(A,B)}(a^0,a^1)=\{b^0,b^1\} \) and \(\Delta_{C(A,B)}(b^0,b^1)=\varnothing.
\)
Thus $\ell_\tau(B) \in \Delta_{C(A,B)}(a^0,a^1)$, which implies that $A\to B$ is an arc of $L_\tau(C(A,B))$, and $r_\tau(B) \in \Delta_{C(A,B)}(a^0,a^1)$, which implies that $B\to A$ is an arc of $R_\tau(C(A,B))$. \(\Delta_{C(A,B)}(b^0,b^1)\) being \(\varnothing
\) excludes the possibility of other arcs being in $L_\tau(C(A,B))$ or $R_\tau(C(A,B))$. \end{proof}

\subsection{Full Proof of 
Lemma~\ref{lem:construction-structure}}\label{sec:lemma5full}

The statement of Lemma~\ref{lem:construction-structure} is reproduced below for reference.

For every clause $C \in \mathcal C$, define \(V_C = \bigcup_{P \in \mathcal P_C} P\). That is, $V_C$ consists of all of the vertices in the pair classes of $\mathcal P_C$. Define $\mathcal M = \bigcup_{z \in Z} \{M_z\}$ and  $V_\mathcal M = \bigcup_{z \in Z}M_z$. The following statements hold: \begin{enumerate}
\item For every $C \in \mathcal C$, $G_I[V_C] = N(\mathcal P_C)$. The restriction of $\chi_I$ to $V_C$ has color classes precisely the members of $\mathcal P_C$. Moreover, $G_I - V_\mathcal M = \bigsqcup_{C \in \mathcal C} N(\mathcal P_C)$.
\item For every $C \in \mathcal C$ and $i \in \mathbb Z_3$, $G_I[M_{z_{C,i}} \cup X_{C,i} ]= E(M_{z_{C,i}}, X_{C,i})$. Furthermore, every edge of $G_I$ which is incident to a vertex in $M_z$ has the form \(m_z^1x_{C,i}^0\) for some $C,i$ for which $z_{C,i} = z$.
\item Let $\tau : \mathcal P_I \to \{0,1\}$ be any orientation assignment, and, for each $C \in \mathcal C$, define $\tau_C = \tau|_{\mathcal P_C}$. Then $L_\tau(G_I) - \mathcal M = \bigsqcup_{C\in\mathcal C} L_{\tau_C}(N(\mathcal P_C))$ and $R_\tau(G_I) - \mathcal M = \bigsqcup_{C\in\mathcal C} R_{\tau_C}(N(\mathcal P_C))$.
\end{enumerate}

\begin{proof}[Full Proof of Lemma~\ref{lem:construction-structure}]

 We first turn attention to showing item~1. 
 \eqref{eq:construction} gives the construction $G_I = (\bigcup_{C \in \mathcal C} N(\mathcal P_C)) \cup (\bigcup_{(C,i) \in \mathcal C \times \mathbb{Z}_3}E(M_{z_{C,i}},X_{C,i}))$.
 The vertex sets of the NAE  gadgets are pairwise disjoint, and every edge of $E(M_{z_{C,i}},X_{C,i})$ has one endpoint in $M_{z_{C,i}}$. Therefore no edge in $E(M_{z_{C,i}},X_{C,i})$ either has both endpoints in $V_C$ or joins $V_C$ to $V_{C'}$ for $C, C' \in \mathcal C$. Therefore we have $G_I[V_C]=N(\mathcal P_C)$ for every $C$ and $G_I-V_{\mathcal M}
  =
  \bigsqcup_{C\in\mathcal C} N(\mathcal P_C)$. The assertion with respect to $\chi_I$ which item~1 makes is justified by definition. This proves item~1 entirely.

  For item~2, fix $C \in \mathcal C$ and $i \in \mathbb Z_3$. The gadget used to construct $G_I$ which contributes an edge between $M_{z_{C,i}}$ and $X_{C,i}$ is $E(M_{z_{C,i}},X_{C,i})$, whose unique edge is $m_{z_{C,i}}^1x_{C,i}^0$. All other equality gadgets used in \eqref{eq:construction} which involve $M_{z_{C,i}}$ connect it to a different pair class than $X_{C,i}$. Thus \(G_I[M_{z_{C,i}}\cup X_{C,i}]
  =
  E(M_{z_{C,i}},X_{C,i})\). The assertion of item~2 regarding edges of $G_I$ incident to a vertex in $M_z$ follows from the construction just detailed. This proves item~2.

  For item~3, let $C \in \mathcal C$ be arbitrary and fix distinct pair classes $P,Q \in \mathcal P_C$. By item~1, we have $G_I[V_C] = N(\mathcal P_C)$. Thus every vertex in $Q$ has the same adjacency relationship to the vertices in $P$ in $N(\mathcal P_C)$ as it has in $G_I$. Since $\tau_C$ is the restriction of $\tau$, we have $\ell_{\tau_C}(Q) = \ell_\tau(Q)$ and $r_{\tau_C}(Q) = r_\tau(Q)$. By the definitions of the left and right digraphs, $L_\tau(G_I)[\mathcal P_C] = L_{\tau_C}(N(\mathcal P_C))$ and $R_\tau(G_I)[\mathcal P_C] = R_{\tau_C}(N(\mathcal P_C))$.

  Now let $C, C' \in \mathcal C$ be distinct clauses, and let $P \in \mathcal P_C$ and $Q \in \mathcal P_{C'}$. By item~1, no edge of $G_I$ joins together a vertex of $P$ to a vertex of~$Q$; thus neither endpoint of either pair class distinguishes the two vertices of the other pair class. This implies that there is no arc in either direction between $P$ and $Q$ in any of $L_\tau(G_I)$ or $R_\tau(G_I)$. 

  After deleting nodes of the form $M_z$ from both digraphs, the remaining vertices are partitioned by the families $\mathcal P_C$, and there are no arcs between distinct such families by the preceding paragraph. This establishes item~3.
\end{proof}

\subsection{Full Proof of Theorem~\ref{thm:np-complete}}\label{sec:full-np-complete}
Theorem~\ref{thm:np-complete} states that \textsc{Lettericity} is \textsf{NP}-complete. 
\begin{proof}[Full proof of Theorem~\ref{thm:np-complete}]
Fix some method of encoding an instance $(G,k)$ of \textsc{Lettericity} as a bit string such that $|V(G)|$ is polynomially bounded by the length of the encoding, arbitrary bit strings can be recognized in polynomial time as being instances of this encoding scheme, and $V(G)$, $E(G)$, and $k$ can be accessed in polynomial time. Because the letter function in a letter representation may be assumed surjective per the discussion to that effect in the preliminaries, a letter representation $x = (\prec, \lambda: V(G) \to \Sigma, D)$ of  $G$ using at most $k$ letters can be encoded in size polynomial in $|(G,k)|$ by, for $n \coloneqq |V(G)|$,  using a length-$n$ ordered list to represent $\prec$, a length-$n$ ordered list whose $i$th entry consists of the letter to which the $i$th vertex is mapped under $\lambda$, and an $n \times n$ boolean matrix $M$ for which $M_{a,b} = 1$ if and only if $(a,b) \in D$.

Consider the algorithm $V((G,k),x)$, which first checks whether its inputs $(G,k)$ and $x$ are bit strings which conform to the respective encoding schemes fixed above. If this condition is not satisfied, $V$ outputs $0$. If it is satisfied, $V$ then checks whether $x = (\prec, \lambda: V \to \Sigma, D)$ satisfies $|\Sigma| \le k$. Finally, $V$ checks whether $x$ is a valid letter representation of $G$ using at most $k$ letters by verifying that for all $u \prec v$, \(uv \in E(G) \iff (\lambda(u), \lambda(v)) \in D\). If any of these conditions aren't satisfied, $V$ halts and outputs $0$; if the last check it performs comes back in the affirmative, $V$ outputs $1$. 

$V$ runs in polynomial time; in particular, checking  that for all $u \prec v$, \(uv \in E(G) \iff (\lambda(u), \lambda(v)) \in D\), involves iterating through $n \choose 2$ items, for each of which items a polynomial amount of time is spent to access $E(G)$ and $D$. 

If $(G,k)$ is a \textsc{yes}-instance of \textsc{Lettericity}, there exists a letter representation of $G$ whose alphabet  is of size at most $k$, and thus there exists a bit string conforming to the fixed encoding convention for certificate instances which induces $V$ to be $1$. Conversely, if there exists some $x$ for which $V((G,k), x) = 1$, then all of the checks which $V$ performs pass when it takes $(G,k)$ and $x$ as input, and thus $x$ is a valid letter  representation of $G$ using at most $k$ letters; therefore $(G,k)$ is a \textsc{yes}-instance. 
Thus $V$ is an \textsf{NP}-verifier for \textsc{Lettericity}.

Combining membership in \textsf{NP} with \textsf{NP}-hardness gives the desired statement. \end{proof}

\subsection{Full Proof of Theorem~\ref{thm:eth-lowerbound}}\label{sec:full-cor-eth-lowerbound}

We here restate Theorem~\ref{thm:eth-lowerbound} in terms of $H$ and $K$ instead of $G$ and $k$ respectively to not conflict with the variables $G$ and $k$ which are used internally within the reduction $R$ and which are referenced in the below proof. 

Unless the Exponential Time Hypothesis is incorrect, there cannot exist a deterministic algorithm which, given a graph $H$ on $h$ vertices and a positive integer~$K$, decides whether $\ell(H) \le K$  in time $2^{o(h)}$, even when $h = 6K$.
\begin{proof}[Full proof of Theorem~\ref{thm:eth-lowerbound}]

Assume the ETH. Under this assumption, there cannot exist a deterministic algorithm which decides \textsc{Monotone-NAE3-SAT} in $2^{o(m)}$ time, where $m$ is the number of clauses in its input instances of \textsc{Monotone-NAE3-SAT}~\cite[Lemma~5.4]{KnoMas2018}.

We now show that the reduction $R$ from  \textsc{Monotone-NAE3-SAT} to \textsc{Lettericity} formed by composing the reductions in Sections~\ref{sec:sattocolor} and~\ref{sec:colortoletter} in the appropriate way can be trivially modified in a way that preserves its correctness and results in it, upon taking as input an instance of \textsc{Monotone-NAE3-SAT} with $n'$ variables and $m$ clauses, outputting a graph with $6K$ vertices for $K =  \Theta(m)$. 

The modification alluded to is but a post-processing step. Let $I$ be an arbitrary instance of \textsc{Monotone-NAE3-SAT} with $n'$ variables and $m$ clauses, and let $(G, \chi : V(G) \to [k])$ be the \CE instance produced by the reduction in Section~\ref{sec:sattocolor} upon taking $I$ as input. Let $(H,K)$ be the output of $R$ upon taking $I$ as input. We may assume that $m \ge 1$, since $I$ is trivially satisfiable if it has no clauses and $R$ can map it to a \textsc{yes}-instance in this case.

Consider the postprocessing step which constructs a new graph $H'$ from $H$ by, for each $v \in V(G)$, arbitrarily deleting vertices in $X_v$ until exactly three vertices remain in it and, for each $i \in [k]$ and $b \in \{0,1\}$, arbitrarily deleting vertices in $M_i^b$ until exactly three vertices remain in it. This postprocessing step can be accomplished, as there are at least $3$ vertices in $X_v$ and $M_i^b$: $m \ge 1$ implies that $k \ge 1$, as the reduction in Section~\ref{sec:sattocolor} sets $k = n' + 6m$; $k \ge 1$ implies $N \ge 3$ as the reduction in Section~\ref{sec:colortoletter} sets $N = 2k + 1$.

Now for any graph $\Gamma$, two distinct vertices $u, v \in V(\Gamma)$ satisfying $\Delta_\Gamma(u,v) = \emptyset$ are called twins. A twin class is a maximal subset of $V(\Gamma)$ for which any pair of distinct vertices in it are twins  (i.e., a set $S \subseteq V(\Gamma)$ of pairwise twins in $V(\Gamma)$ for which there is no $S' \subseteq V(\Gamma)$ such that $S \subsetneq S'$ and the vertices in $S'$ are pairwise twins). For each of the $X_v$ or $M_i^b$, by construction any pair of vertices in it are twins. Thus whenever any of the $X_v$ or $M_i^b$ contain at least 4 vertices, they are part of a twin class of size at least $4$. Therefore, iteratively applying~\cite[Lemma~17]{AleKLZ2026} implies that $\ell(H') = \ell(H)$. As such, adding the discussed post-processing step to $R$ results in a valid Karp-reduction from \textsc{Monotone-NAE3-SAT} to \textsc{Lettericity}, as deleting the vertices of $H$ is of course polynomial as a polynomial number of vertices are deleted. 

The postprocessing step does not at all modify $K$. Thus, due to the fact that the reduction in Section~\ref{sec:colortoletter} sets $K = 2k$ and the reduction in Section~\ref{sec:sattocolor} sets $k = n' + 6m$, we have $K = 2k = 2(n'+6m) = \Theta(n'+m)$. Because we may assume $n' \le 3m$ as unused variables in~$I$ can be deleted, we have $K = \Theta(m)$.

We also have, by construction, $|V(H')| = 3(|V(G)| + 2k) = 3(K + K) = 6K$.

Now if Theorem~\ref{thm:eth-lowerbound} was false, $\ell(H') \le K$ could be decided in time $2^{o(|V(H')|)}$, meaning that $I$ could be decided in time $2^{o(m)}$. Because $I$ was arbitrary and the algorithm to decide \textsc{Lettericity} just alluded to does not depend on $I$, this gives a   $2^{o(m)}$ deterministic algorithm to decide \textsc{Monotone-NAE3-SAT}, contradicting the bound on that problem cited at the beginning of this proof.
\end{proof}

\subsection{Proof of~\Cref{thm:fpt-dual}}\label{sec proof fpt}
In order to present our result, we use the following notation.
Let~$X$ and~$Y$ be sets with~$X\subseteq Y$.
We say that a linear order~$\prec_X$ of~$X$ is~\emph{a restriction of} the linear order~$\prec_Y$ of~$Y$ if for each~$x_1,x_2\in X$, $x_1 \prec_X x_2 \Leftrightarrow x_1 \prec_Y x_2$.
In this case, we may also say that~$\prec_Y$~\emph{extends}~$\prec_X$.

\newcommand{\mC}{\mathcal{C}}
\begin{lemma}\label{single classes are trivial}
Let~$G=(V,E)$ be a graph and let~$\chi\colon V\to [k]$~be a coloring of~$G$, such that at least one color class has a size of~$1$.
Moreover, let~$\mC$ be the collection of all color classes of size at least~$2$ under~$\chi$, let~$W \coloneqq  \bigcup_{C\in \mC} C$, and let~$\prec_W$ be a linear order of~$W$.
Then, there is a decoder~$D$ and a linear order~$\prec$ of~$V$ that extends~$\prec_W$ for which~$(\prec,\chi,D)$ is a letter representation of~$G$ if and only if
for each~$x\in V\setminus W$, there is some linear order~$\prec_x$ of~$W\cup \{x\}$ that extends~$\prec_W$ and for which there is a decoder~$D_x$, such that~$(\prec_x, \chi_{W\cup\{x\}},D_x)$ is a letter representation of~$G[W\cup\{x\}]$.
\end{lemma}
\begin{proof}
The first direction is trivial:
If there is a decoder~$D$ and a linear order~$\prec$ of~$V$ that extends~$\prec_W$ for which~$(\prec,\chi,D)$ is a letter representation of~$G$, then for each~$x\in V\setminus W$, $(\prec_x, \chi_{W\cup\{x\}},D_x)$ is a letter representation of~$G[W\cup\{x\}]$ for~$\prec_x$ being the restriction of~$\prec$ to~$W\cup \{x\}$ (and thus extends~$\prec_W$) and $D_x$ being the restriction of~$D$ to the colors used by~$W\cup \{x\}$.

Thus, consider the opposite direction.
For each~$x\in V\setminus W$, let~$\prec_x$ be a linear order of~$W\cup \{x\}$ that extends~$\prec_W$ and for which there is a decoder~$D_x$, such that~$(\prec_x, \chi_{W\cup\{x\}},D_x)$ is a letter representation of~$G[W\cup\{x\}]$.
Now let~$\prec$ be an arbitrary linear order of~$V$ that extends all the linear orders~$\prec_x$.
Such a linear order exists, as all the linear orders~$\prec_x$ extend~$\prec_W$ and pairwise share no other vertex.
We show that there is a decoder~$D$ such that~$(\prec,\chi,D)$ is a letter representation of~$G$.
By \autoref{lem:orderofdistinguishers}, it thus suffices to show that for any two vertices~$u,v$ with~$\chi(u) = \chi(v)$ (satisfying, without loss of generality, $u\prec v$), each vertex~$q\in \Delta_G(u,v)$ fulfills~$u\prec q\prec v$.
This is fulfilled, since~$\prec$ extends all the linear orders~$\prec_x$, for each~$x\in V\setminus W$, (\ref{eq:orderofdistinguishers}) is fulfilled, as there is a decoder~$D_x$, such that~$(\prec_x, \chi_{W\cup\{x\}},D_x)$ is a letter representation of~$G[W\cup\{x\}]$, and each vertex~$q\in V \supseteq \Delta_G(u,v)$ is part of some set~$W\cup \{x\}$ for some~$x\in V \setminus W$.
\end{proof}


With this result, we are now ready to prove~\Cref{thm:fpt-dual}.
\begin{proof}[Proof of~\Cref{thm:fpt-dual}]
Let~$(G = (V,E),k)$ be an instance of~\LET.
First, we present the XP algorithm.
The idea for the algorithm is to brute-force all possible~$k$-colorings of~$G$ (up to isomorphism) and for each such coloring~$\chi$, all the linear orders of the vertices outside of size-1 color classes.
As we will show, there are~$(2\duall)!\cdot n^{2\cdot \duall} \in n^{\Oh(\duall)}$ such choices for the coloring and the partial linear order.
For each such choice~$(\chi,\prec_W)$, where~$W$ denotes the vertices for which the color class has size at least~$2$, we then check whether there is a decoder~$D$ and a linear order~$\prec$ of~$V$ that extends~$\prec_W$ for which~$(\prec,\chi,D)$ is a letter representation of~$G$.
This is correct, since~$G$ has lettericity at most~$k$ if and only if there is a~$k$-coloring~$\chi$ for~$G$, a linear order~$\prec$ of~$V$, and a decoder~$D$, such that~$(\prec,\chi,D)$ is a letter representation of~$G$.

We check this property as follows:
If~$W = V$, we only need to check whether (\ref{eq:orderofdistinguishers}) holds due to~\Cref{lem:orderofdistinguishers}.
This can clearly be done in polynomial time by evaluating (\ref{eq:orderofdistinguishers}).
Otherwise, there is at least one color class of size~$1$.
By~\Cref{single classes are trivial}, this can be boiled down to checking whether for each~$x\in V \setminus W$, there is some linear order~$\prec_x$ of~$W\cup \{x\}$ that extends~$\prec_W$ and for which there is a decoder~$D_x$, such that~$(\prec_x, \chi_{W\cup\{x\}},D_x)$ is a letter representation of~$G[W\cup\{x\}]$. 
As there are at most~$n$ choices for~$x$ and at most~$n$ choices for a linear order~$\prec_x$ that extends~$\prec_W$, this property can also be checked in polynomial time by evaluating (\ref{eq:orderofdistinguishers}).

Thus, the running time of this algorithm is upper-bounded by a polynomial function in~$n$ times the number of choices for~$(\chi,\prec_W)$.
We start by showing that there are at most~$n^{2\duall}$ many choices for~$\chi$.
As we are asking for a coloring with~$k = n-\duall$ colors, at least~$k-\duall = n-2\duall$ color classes can be assumed to have size~1, as we do not need to consider empty color classes.
The extreme case is achieved by having, when $2\duall \le n$, $\duall$ many color classes of size exactly~$2$.
We can enumerate such colorings as follows:
Enumerate all vertex sets of size~$\duall$.
For each such vertex set~$X$, let~$Y \coloneqq  V\setminus X$, and enumerate all functions~$\pi\colon X\to Y$.
Such a function should express that~$x$ will be in the same color class as~$\pi(x)$, and that all vertices of~$Y$ will be in pairwise distinct color classes.
That is, each such choice~$(X,Y,\pi)$ corresponds to the coloring~$\chi\colon V \to Y$ (assuming without loss of generality that~$Y = [k]$) where~$\chi|_Y$ is the identity, and where~$\chi(x) = \pi(x)$ for all~$x\in X$.
This enumerates all colorings of~$G$ with exactly~$k$ color classes (up to isomorphism).
Moreover, as there are at most~$n^{\duall}$ many choices for each of~$X$ and~$\pi$, we conclude that there are~$n^{2\duall}$ many choices for~$\chi$.

Next, we discuss the choices for~$\prec_W$.
Again, as we assume that there are at least~$n-2\duall$ many color classes of size~$1$ under~$\chi$, the set~$W$ of all vertices in color classes of size at least~$2$ has size at most~$2\duall$.
This implies that there are only~$(2\duall)!$ choices for~$\prec_W$ for each fixed coloring~$\chi$.
Hence, in combination, there are~$(2\duall)!\cdot n^{2\cdot \duall} \in n^{\Oh(\duall)}$ many cases to consider for~$(\chi,\prec_W)$ (which can be enumerated in the same asymptotic running time).

We now show that this running time also implies the claimed \textsf{FPT} running time.
This follows by the fact that the lettericity of a graph is never larger than~$n-\lfloor\frac{1}{2}\log_2(n)\rfloor$~\cite{ManVat2024}.
Hence, if $k \geq n-\lfloor\frac{1}{2}\log_2(n)\rfloor$, we can trivially answer \textsc{yes} in $\Oh(n)$~time.
Otherwise, $\duall > \lfloor\frac{1}{2}\log_2(n)\rfloor$, which implies that the above-stated \textsf{XP} running time is $n^{\Oh(\duall)}\subseteq 2^{\Oh(\duall^2)}$~time.
By checking $\duall \leq  \lfloor\frac{1}{2}\log_2(n)\rfloor$ first, this gives a running time of $2^{\Oh(\duall^2)} + \Oh(n)$~time.
\end{proof}

\subsection{Proof of Theorem~\ref{thm:word-extension-np-complete}}\label{sec:word-extension-np-complete}

\begin{proof}[Proof of Theorem~\ref{thm:word-extension-np-complete}]
For membership in \textsf{NP}, let $(G',w)$ be an instance of \WE, where $n = |V(G')|$ and $w = (w_1, \ldots, w_n)$. Let $\Sigma$ be exactly the elements of $w$.
A certificate is an enumeration $v_1, \ldots, v_n$ of $V(G')$ and some $D$ for $D \subseteq \Sigma^2$; this certificate is polynomial-size as $|\Sigma| \le n$. It can be verified that  $(G',w)$ is a \textsc{yes}-instance of \WE by, after verifying that $v_1, \ldots, v_n$ is an enumeration of $V(G')$ and that $D$ is indeed a subset of $\Sigma^2$, defining a linear order $\prec$ of $V(G')$ such that $v_1 \prec v_2 \prec \cdots \prec v_n$ and a function $\lambda$ such that for all $i \in [n]$ $\lambda(v_i) = w_i$ and then verifying, using an \textsf{NP}-verifier for \textsc{Lettericity}, that $(\prec, \lambda, D)$ is a letter representation of $G'$.

For \textsf{NP}-hardness, we reduce from \textsc{3-Partition} with the input integers encoded in unary, which is \textsf{NP}-complete
~\cite[Problem SP15]{GarJoh79}. This problem takes as input unary-encoded positive integers $a_1,\ldots,a_{3m}$, and $B$ such that $\sum_{i=1}^{3m}a_i=mB$ and, for $i \in [3m]$, $\frac B4<a_i<\frac B2$, and asks whether $[3m]$ can be partitioned into $m$ disjoint sets $J_1, \ldots, J_m$ such that for each $t \in [m]$, $\sum_{i \in J_t} a_i = B$.

Let $(a_1,\ldots,a_{3m}, B)$, encoded in unary, be an arbitrary \textsc{3-Partition} instance. Set $H = \bigsqcup_{i = 1}^m K_B$. Define the alphabet $\Sigma=\{\sigma_1,\ldots,\sigma_{3m}\}$, where the $\sigma_i$ are pairwise distinct letters. Now define $W=\sigma_1^{a_1}\Vert\cdots\Vert\sigma_{3m}^{a_{3m}}$, where  $\sigma_i^{a_i}$ is the length-$a_i$ tuple for which each entry is $\sigma_i$, and $\Vert$ is tuple concatenation, in that, for example, $(a,b,c) \Vert (d,e,f) = (a,b,c,d,e,f)$. The reduction we give outputs $(H,W)$. The entire construction is polynomial-time given the unary encoding ($H$ has $mB$ vertices and $m \binom B2$ edges).

For correctness, fix an enumeration  $C_1, \ldots, C_m$ of the vertex sets of the connected components of $H$, such that $H[C_t]$ is a copy of $K_B$. Assume first that $(a_1,\ldots,a_{3m}, B)$ is a \textsc{yes}-instance of \textsc{3-Partition}. Then there exist disjoint sets $J_1,\ldots,J_m$ which partition $[3m]$ and for which, for each $t \in [m]$, $\sum_{i \in J_t} a_i = B$. For each $t \in [m]$ and each $i \in J_t$, choose a subset $U_i \subseteq C_t$ such that the sets $U_i: i \in J_t$ partition $C_t$ and, for every $i \in J_t$, $|U_i| = a_i$ (it is possible to make these choices as $\sum_{i \in J_t}a_i = B = |C_t|$). Thus $U_1, \ldots, U_{3m}$ partition $V(H)$. Let $\prec$ order the vertices of $H$ such that for $U_1, \ldots, U_{3m}$, every vertex in a set is ordered before every vertex in a set appearing later (how $\prec$ orders vertices in the same $U_i$ is not relevant). Define $\lambda$ as mapping each vertex in $U_i$ to $\sigma_i$. Define $D=\bigcup_{t=1}^{m}
      \{(\sigma_i,\sigma_j):i,j\in J_t\}$. Vertices $u,v \in V(H)$ with $u \prec v$ satisfy $(\lambda(u), \lambda(v)) \in D$ exactly when they belong to the same clique $K_B$ in $H$, and thus $(\prec, \lambda, D)$ is a letter representation of $H$. 

      Because for every $i, j \in [3m]$ for which $i< j$, every vertex in $U_i$ is ordered by $\prec$ to be before every vertex of $U_j$; and because moreover, we have $|U_i| = a_i$ and $\lambda(v) = \sigma_i$ for every $v \in U_i$; it follows that  if $V(H)$ is enumerated as $v_1  \prec v_2 \prec \cdots \prec v_{mB}$, and if  the elements of $W$ are indexed as $(w_1, \ldots, w_{mB})$, then for each $r \in [mB]$, $\lambda(v_r) = w_r$. Hence, $(H,W)$ is a \textsc{yes}-instance of \WE.

    Conversely, assume that $(\prec, \lambda, D)$ is a letter representation of $H$ for which after enumerating $V(H)$ as $v_1 \prec v_2  \prec\cdots \prec v_{mB}$, for all $i \in [mB]$ $\lambda(v_i) = w_i$. For each $i \in [3m]$, define $S_i = \{v \in V(H):\lambda(v) = \sigma_i\}$; due to how $(\prec, \lambda, D)$ certifies that $(H,W)$ is a \textsc{yes}-instance of \WE, we have that $|S_i| = a_i$ and, for $u,v \in S_i$ where $u$ and $v$ are distinct, $I_{\prec}(u,v) \subseteq S_i$. 
    
    Each $S_i$ is completely contained in a single connected component of $H$. To see this, let $u,v \in S_i$. Because all occurrences of $\sigma_i$ in $W$ form a contiguous block, every vertex lying between $u$ and $v$ in $\prec$ must also belong to $S_i$. Thus if $u$ and $v$ are consecutive in $S_i$, in that there does not exist a vertex in $S_i$ lying between $u$ and $v$ in $\prec$, then no vertex of $H$ lies between them either. As $\lambda(u)=\lambda(v)$, Lemma~\ref{lem:orderofdistinguishers} implies that $u$ and $v$ are twins, so $u$ and $v$ are either in the same clique or isolated vertices; and the latter case is impossible, because each of the connected components of $H$ is of size $B>1$. 

For each $t \in [m]$, define $J'_t = \{i \in [3m]: S_i \subseteq C_t\}$. The $S_i$ partition $V(H)$ and each lies in exactly one of its connected components; thus $J'_1, \ldots, J'_m$ partition $[3m]$, and, for each $t \in [m]$, $\sum_{i \in J'_t} a_i = \sum_{i \in J'_t} |S_i| = |C_t| = B$. Therefore $(a_1,\ldots,a_{3m}, B)$ is a \textsc{yes}-instance of \textsc{3-Partition}. \end{proof}

\subsection{Proof of~\Cref{thm:decoder-extension-np-complete}}\label{sec DE hard}

\newcommand{\mP}{\mathcal{P}}
\begin{proof}[Proof of~\Cref{thm:decoder-extension-np-complete}]
The containment in~\textsf{NP} is trivial, as we can check in polynomial time whether~$(\prec,\chi,D)$ is a letter representation of~$G$, and the encodings of~$\prec$ and~$\chi$ are polynomial in the length of the input.

We now show the hardness.
To this end, we reduce from~\CE which we proved to be \textsf{NP}-hard even if each color class is a pair class (see \Cref{thm:coloring-extension-npcomplete}).
In fact, the hardness for~\CE even holds when the edge set between any two pair classes~$P,Q$ is either (i)~empty, (ii)~the edge set of an equality gadget, or (iii)~the edge set of a containment gadget (see~\Cref{thm:sattocolor}). 
Let~$I\coloneqq (G=(V,E),\Sigma,\chi)$ be such an instance of~\CE.
We extend the graph and define a decoder~$D$ that enforces to choose exactly the color classes forced by~$\chi$.
Note that the number of colors assigned by~$\chi$ is~$n/2$, where~$n$ is the number of vertices of~$G$, as each color class has size~$2$. Let~$P_i$ denote the~$i$th color class under~$\chi$ for some arbitrary but fixed ordering of the color classes.

We now define an equivalent instance~$I'$ of~\DE (consisting of a graph~$G'$, an alphabet~$\Sigma'$, and a decoder~$D$). We start off with $G' = G$, $\Sigma' = \emptyset$, and $D = \emptyset$; we then add to these objects as follows. Let~$P_i = \{p^0,p^1\}$ be a pair class under~$\chi$ and let~$\alpha_i$ denote the letter with~$\chi^{-1}(\alpha_i) = P_i$.
We add the two letters~$\alpha^0_i$ and~$\alpha^1_i$ to~$\Sigma'$ and the decoder-words~$\{\alpha^0_i\alpha^1_i, \alpha^1_i\alpha^0_i\}$ to~$D$ but we do~\emph{not} add the decoder-words~$\alpha^0_i\alpha^0_i$ or~$\alpha^1_i\alpha^1_i$ to~$D$.
The idea is that each vertex of~$P_i$ should be free to choose any of the letters~$\alpha^0_i$ or~$\alpha^1_i$ under~$D$, but this choice should be the same for both vertices of~$P_i$, which will then force the set~$P_i$ to be a color class under each solution for~$I'$.
The consistency in the choice comes from the fact that each color class~$P_i$ in~$I$ is an independent set, and, no matter the linear order, $D$ would turn~$P_i$ into a clique if the two vertices choose different letters from~$\{\alpha^0_i,\alpha^1_i\}$.
The meaning of both vertices choosing~$\alpha^0_i$ under some coloring~$\chi'$ is that, in each linear order~$\prec$ for which~$(\prec,\chi',D)$ is a letter representation of~$G'$, $p^0 \prec p^1$.
Analogously, if both vertices chose~$\alpha^1_i$ under some coloring~$\chi'$, this should imply $p^1 \prec p^0$.
We ensure this by defining the decoder-word for the respective color classes of each equality and each containment gadget as follows:

Let~$P_y=\{y^0,y^1\}$ and~$P_z=\{z^0,z^1\}$ be distinct color classes for which~$E(P_y,P_z)$ is an equality gadget in~$G$.
Recall that this means that~$G[P_y,P_z]$ only contains the edge~$y^1z^0$.
Then, we add the following decoder-words to~$D$:
\begin{itemize}
\item $\alpha^d_y\alpha^{1-d}_z$ and~$\alpha^{1-d}_z\alpha^d_y$ for each~$d\in \{0,1\}$ (which ensures that the following choices are impossible as they would create all possible four edges between~$P_y$ and~$P_z$: (i)~$\alpha^0_y$ for~$P_y$ and simultaneously~$\alpha^1_z$ for~$P_z$, or (ii)~$\alpha^1_y$ for~$P_y$ and simultaneously~$\alpha^0_z$ for~$P_z$), and
\item $\alpha^0_z\alpha^0_y$ and~$\alpha^1_y\alpha^1_z$ (which allows that a linear order~$\prec$ restricted to~$P_y\cup P_z$ allows~$y^0 \prec z^0 \prec y^1 \prec z^1$ and~$z^1 \prec y^1 \prec z^0 \prec y^0$, respectively).
\end{itemize}

Similarly, let~$P_y=\{y^0,y^1\}$ and~$P_z=\{z^0,z^1\}$ be distinct color classes for which~$C(P_y,P_z)$ is a containment gadget.
Recall that this implies that~$G[P_y,P_z]$ only contains the edges~$y^1z^0$ and~$y^1z^1$.
Then for each~$d\in \{0,1\}$, we add  the decoder-words~$\alpha^1_y\alpha^d_z$ and~$\alpha^d_z\alpha^0_y$ to~$D$.

For the correctness of the reduction, we also need a way to ensure that for each color class~$P_i=\{p^0,p^1\}$, both~$p^0$ and~$p^1$ are only allowed to be colored with~$\alpha^0_i$ or~$\alpha^1_i$.
We ensure this as follows:
For each~$k\in [1,n/2 + 1]$, we add to $G'$ a clique~$S_k$ of size~$2n$ and a new letter~$\beta_k$ to~$\Sigma'$ with the decoder-word~$\beta_k\beta_k$ to~$D$.
Then, for each~$i\in [1,n/2]$, we add edges to $G'$  such that both vertices of color class~$P_i$ are adjacent to all vertices of exactly~$i+1$ arbitrary but fixed  newly introduced cliques $S_k$.
For each such clique~$S_k$ for which we added the edges towards~$P_i$, we add the decoder-words~$\{\alpha^d_i\beta_k,\beta_k\alpha^d_i \mid d\in \{0,1\}\}$ to~$D$.

This completes the definition of~$I'$.
Observe that~$|\Sigma'|  = 2\cdot |\Sigma| + \frac{n}{2} + 1 < 2n$, as we can assume that~$n > 2$.
Note that this implies that~$\Sigma'$ has size less than any newly introduced clique~$S_k$.
We now show that~$I$ is a \textsc{yes}-instance of~\CE if and only if~$I'$ is a \textsc{yes}-instance of~\DE.

$(\Leftarrow)$
Let~$\prec'$ be a linear order of~$V(G')$ and let~$\chi'$ be a coloring of~$V(G')$, such that~$(\prec',\chi',D)$ is a letter representation of~$G'$.
Our first goal is to show that for each color class~$P_i$ under~$\chi$, both vertices of~$P_i$ receive the same color from~$\{\alpha^0_i,\alpha^1_i\}$ under~$\chi'$.

\begin{claim}
For each color class~$P_i$ under~$\chi$, both vertices of~$P_i$ receive the same color from~$\{\alpha^0_i,\alpha^1_i\}$ under~$\chi'$.
\end{claim}
\begin{claimproof}
We first observe that (without loss of generality) for each newly added clique~$S_k$, all vertices of~$S_k$ receive the same color, and that this color is of the form~$\beta_j$.

Let~$S_k$ be a newly added clique.
By construction, $S_k$ has size~$2n > |\Sigma'|$.
This implies that at least two vertices of~$S_k$ receive the same color under~$\chi'$. Let~$\gamma$ be that color.
As all vertices of~$S_k$ have the exact same closed neighborhood, we can thus assume, without loss of generality, that~$S_k \subseteq \chi'^{-1}(\gamma)$.
Moreover, as~$S_k$ is a clique, it is necessary that~$\gamma\gamma$ is part of the decoder~$D$.
This is only the case for colors of the form~$\beta_j$.
Moreover, distinct cliques must receive distinct colors, as there are no edges between different cliques, and~$\beta_p\beta_q$ is only in~$D$ if~$p=q$.

Based on this fact, we can now show that for each color class~$P_i=\{p^0,p^1\}$ of~$\chi$, both vertices receive the same color from~$\{\alpha^0_i,\alpha^1_i\}$ under~$\chi'$.
Let~$d\in \{0,1\}$ and consider~$\gamma \coloneqq  \chi'(p^d)$.
This vertex has by construction edges towards exactly~$i+1\geq 2$ cliques.
This implies that~$\chi'(p^d)$ is not of the form~$\beta_j$, as otherwise, $p^d$ could have only incident edges with a single clique.
Moreover, as each such clique has a distinct color of the form~$\beta_j$ under~$\chi'$, the decoder~$D$ has to contain a word between~$\gamma$ and~$\beta_j$ for exactly~$i+1$ many values of~$j$.
By construction, this is only the case for~$\gamma \in \{\alpha^0_i,\alpha^1_i\}$.
It thus remains to show that~$\chi'(p^0) = \chi'(p^1)$.
This is the case, since~$P_i$ is an independent set and~$D$ contains both words~$\alpha^0_i\alpha^1_i$ and~$\alpha^1_i\alpha^0_i$, which would otherwise enforce an edge between~$p^0$ and~$p^1$.
\end{claimproof}

We are now ready to show that~$I$ is a \textsc{yes}-instance of~\CE.
This immediately follows from the fact that~$\chi'|_{V(G)}$ and~$\chi$ have the same color classes, $G'[V(G)] = G$, and $(\prec',\chi',D)$ is a letter representation of~$G'$.

$(\Rightarrow)$
Let~$\prec$ be a linear order of~$V$ and let~$D_G$ be a decoder, such that~$(\prec,\chi,D_G)$ is a letter representation of~$G$.
Now, let~$\prec'$ be an arbitrary linear order of~$V(G')$ which extends~$\prec$.
Moreover, let~$\chi'$ be the coloring of the vertices of~$V(G')$ defined as follows:
For each~$k\in [1,n/2 + 1]$, and each vertex~$v$ of clique~$S_k$, we set~$\chi'(v) \coloneqq  \beta_k$.
For each other vertex~$v\in V(G')$, we know that~$v\in V$.
That is, there is some color class~$P_i=\{p^0,p^1\}$ under~$\chi$, such that~$v\in P_i$.
Recall that we denoted the color assigned to~$P_i$ under~$\chi$ by~$\alpha_i$.
We set~$\chi'(v) \coloneqq  \begin{cases}
\alpha^0_i & p^0\prec p^1 \\
\alpha^1_i & p^1\prec p^0
\end{cases}$.
In the following, we will show that~$(\prec',\chi',D)$ is a letter representation of~$G'$.

To this end, we show that for all pairs of distinct vertices $u,v$, we have that
\begin{align}\label{property edges}
uv \in E(G') \Leftrightarrow (\chi'(u)\chi'(v)\in D \land u \prec' v) \lor (\chi'(v)\chi'(u)\in D \land v \prec' u).
\end{align}
Recall the edge set between any two pair classes~$P,Q$ under~$\chi$ is either (i)~empty, (ii)~the edge set of an equality gadget, or (iii)~the edge set of a containment gadget (see~\Cref{thm:sattocolor}).
Thus, to show the equivalence from~\eqref{property edges}, we will use the following fact about the restriction of the linear order~$\prec$ to the vertices of~$P\cup Q$.

\begin{claim}\label{linear orders for gadgets}
Let~$P=\{p^0,p^1\}$ and~$Q=\{q^0,q^1\}$ be distinct pair classes under~$\chi$.
\begin{itemize}
\item If $G[P,Q]$ is an equality gadget~$E(P,Q)$, then \begin{itemize}\item (i)~$p^0\prec q^0 \prec p^1 \prec q^1$, or (ii)~$q^1\prec p^1 \prec q^0 \prec p^0$.\end{itemize}
\item If $G[P,Q]$ is a containment gadget~$C(P,Q)$, then~$Q \subseteq I_\prec(p^0,p^1)$, that is, \begin{itemize}\item (i)~$p^0\prec q^0, q^1 \prec p^1$ or (ii)~$p^1\prec q^0, q^1 \prec p^0$. \end{itemize}
\end{itemize}
\end{claim} 
\begin{claimproof}
Consider the case that $G[P,Q]$ is an equality gadget~$E(P,Q)$.
Since the only edge between~$P$ and~$Q$ is~$p^1q^0$, we have~$p^1 \in \Delta_G(q^0,q^1)$ and~$q^0 \in \Delta_G(p^0,p^1)$. 
\Cref{lem:orderofdistinguishers} thus implies that~$p^1 \in I_\prec(q^0,q^1)$ and~$q^0 \in I_\prec(p^0,p^1)$.
The only way~$\prec$ can fulfill both these properties is if (i)~$p^0\prec q^0 \prec p^1 \prec q^1$, or (ii)~$q^1\prec p^1 \prec q^0 \prec p^0$.

Now consider the case that $G[P,Q]$ is a containment gadget~$C(P,Q)$.
Since the only edges between~$P$ and~$Q$ are~$p^1q^0$ and~$p^1q^1$, we have~$Q \subseteq \Delta_G(p^0,p^1)$. 
\Cref{lem:orderofdistinguishers} thus implies that~$Q \subseteq I_\prec(p^0,p^1)$.
\end{claimproof}

Note that \eqref{property edges} holds directly for all vertex pairs within the same color class under~$\chi'$, as exactly the cliques~$S_k$ receive a color~$\beta_k$ for which~$\beta_k\beta_k\in D$.
By construction, $D$ does not contain any words between colors for which the edge set between the respective color classes is empty.
So it remains to only look at the edges between any two color classes~$P$ and~$Q$ under~$\chi'$ for which~$G'[P,Q]$ contains at least one edge.
This is only the case if (i)~exactly one of the two color classes is a clique~$S_k$ (without loss of generality~$P$), or if (ii)~$P$ and~$Q$ are both color classes under~$\chi$ by definition of~$\chi'$.

In the first case, \eqref{property edges} holds, since by construction, the edges between a clique~$S_k$ and the color class~$Q$ exist if and only if~$\{\alpha^d_i\beta_k,\beta_k\alpha^d_i \mid d\in \{0,1\}\} \subseteq D$ for~$\alpha_i$ being the color assigned to the vertices of~$Q$ under~$\chi$ (so the concrete choice of~$\prec'$ does not matter).
So, consider the second case from now on.
Without loss of generality, let~$P = \{p^0,p^1\}$, $Q = \{q^0,q^1\}$, and let~$\alpha_y$ and~$\alpha_z$ be the colors assigned by~$\chi$ to the vertices of~$P$ and~$Q$, respectively.
By assumption, if there is at least one edge between~$P$ and~$Q$, then the two sets form an equality gadget or a containment gadget.

\smallskip
\noindent
\textbf{Case 1:} Assume that~$G[P,Q]$ is an equality gadget~$E(P,Q)$\textbf{.}
Recall that this implies that~$p^1q^0$ is the only edge between these two sets.
By~\Cref{linear orders for gadgets}, we know that (i)~$p^0\prec q^0 \prec p^1 \prec q^1$, or (ii)~$q^1\prec p^1 \prec q^0 \prec p^0$.
Assume that~$p^0\prec q^0 \prec p^1 \prec q^1$. 
(The other case can be shown symmetrically.)
Since~$\prec'$ extends~$\prec$, we have~$p^0\prec' q^0 \prec' p^1 \prec' q^1$. 
Moreover, by definition of~$\chi'$, this implies that the colors assigned to vertices of~$P$ and~$Q$ under~$\chi'$ are~$\alpha^0_y$ and~$\alpha^0_z$, respectively.
Furthermore, by construction, $D$ contains the word~$\alpha^0_z\alpha^0_y$ and does not contain the word~$\alpha^0_y\alpha^0_z$.
Hence, \eqref{property edges} holds between~$P$ and~$Q$.

 \smallskip
\noindent
\textbf{Case 2:} Assume that~$G[P,Q]$ is a containment gadget~$C(P,Q)$\textbf{.}
Recall that this implies that~$p^1q^0$ and~$p^1q^1$ are the only edges between these two sets.
By~\Cref{linear orders for gadgets}, we get~$Q \subseteq I_\prec(p^0,p^1)$, which implies that~$Q \subseteq I_{\prec'}(p^0,p^1)$, as~$\prec'$ extends~$\prec$.
Assume that~$p^0\prec p^1$.
(The other case can be shown symmetrically.)
Moreover, let~$d\in \{0,1\}$, such that~$q^d\prec q^{1-d}$. 
By definition of~$\chi'$, this implies that the colors assigned to vertices of~$P$ and~$Q$ under~$\chi'$ are~$\alpha^0_y$ and~$\alpha^d_z$, respectively.
Furthermore, by construction, $D$ contains the word~$\alpha^d_z\alpha^0_y$ and does not contain the word~$\alpha^0_y\alpha^d_z$.
Hence, \eqref{property edges} holds between~$P$ and~$Q$.

Concluding, $(\prec',\chi',D)$ is a letter representation of~$G'$.
\end{proof}

\begin{remark}[ETH Bounds for Extension Problems] For \CE, we are able to conclude that under the ETH, there can be no $2^{o(n)}$ algorithm for it when $n$ is the number of vertices in its input graph, even when the given coloring is proper and its color classes each have size two: the reduction in Section~\ref{sec:sattocolor} maps an instance of \NAE with variable set $Z$ and $m$ clauses to an instance of \CE with the properties specified of its coloring and for which its input graph has $2(|Z| + 6m) = \Theta(m)$ vertices (where the $\Theta$ bound holds under the assumption that $|Z|\le 3m$, without loss of generality). As used in the proof of Theorem~\ref{thm:eth-lowerbound}, under the ETH there cannot exist a $2^{o(m)}$ algorithm for \NAE when $m$ is the number of clauses in its input instances~\cite[Lemma~5.4]{KnoMas2018}, whereas there would be such an algorithm for \NAE under the ETH if there was a $2^{o(n)}$ algorithm for \CE under the ETH. 

This ETH lower bound for \CE implies that under the ETH, there can be neither a $2^{o(\sqrt{n})}$ nor a $2^{o(|\Sigma'|)}$ algorithm for \DE, where $n$ is the number of vertices in the graph component of its input instances: the reduction in the proof of Theorem~\ref{thm:decoder-extension-np-complete} maps a \CE instance whose input graph has $n'$ vertices to a \DE instance whose input graph has $n = {n'}^2 + 3n' = \Theta({n'}^2)$ vertices and whose given alphabet $\Sigma'$ is of size $ \frac {3n'} 2 + 1 = \Theta(n') = \Theta(\sqrt{n})$; a $2^{o(\sqrt{n})}$ or $2^{o(|\Sigma'|)}$ algorithm for \DE under ETH thus contradicts the ETH lower bound for \CE discussed in the previous paragraph. 

    As for \WE, although we would suspect ETH-based lower bounds for it similar to \textsc{Lettericity}, we cannot claim this based on the previous proof: the typical size measures of the \textsc{3-Partition} instance, $m$ and $B$, will be multiplied to give the number of vertices in the constructed instance, the number of edges even being worse. Furthermore, \textsc{3-Partition} admits a $2^{o(n)}$ algorithm when $n$ is the sum of its input integers (one can reduce \textsc{3-Partition} to \textsc{Polyomino Packing} with forbidden rotations by creating a rectangle of height~$1$ and width $a_i$ for each input integer $a_i$ with a target rectangle  of height~$m$ and width~$B$, giving a $2^{\mathcal{O}(\sqrt{n} \log n)}$ algorithm for \textsc{3-Partition} from~\cite[Corollary~3]{BodZan2020}); this means that the \WE instances produced by our reduction admit a $2^{o(n)}$ algorithm where $n$ is the number of vertices in their input graphs. As such, a solution in this direction could be to pursue a reduction, perhaps one similar to the reduction we give in the proof of Theorem~\ref{thm:word-extension-np-complete}, to \WE from other packing and scheduling problems with useful ETH lower bounds, such as flow shop scheduling~\cite{JanLanLan2016}. We also refer to the previous discussions of a possible ETH lower bound based on Vatter's \textsf{NP}-hardness proof for \textsc{Lettericity} in \Cref{sec:intro}. 
    
\end{remark}

\section{Supplementary Figures}\label{sec:extra-figures}

Figures~\ref{fig:leftdigraph} and~\ref{fig:wires} further illustrate our main \textsf{NP}-hardness construction.

\begin{figure}[ht]
  \centering
  \begin{tikzpicture}[
      vertex/.style={
        circle,
        draw,
        minimum size=7mm,
        inner sep=0pt,
        font=\small
      },
      pvertex/.style={vertex,fill=blue!15},
      qvertex/.style={vertex,fill=orange!20},
      rvertex/.style={vertex,fill=green!15},
      pairclass/.style={
        draw,
        dashed,
        rounded corners,
        inner sep=5pt
      },
      graphedge/.style={
        line width=0.9pt
      },
      dagnode/.style={
        circle,
        draw,
        minimum size=8mm,
        inner sep=0pt,
        font=\small
      },
      dagedge/.style={
        -{Latex[length=2.2mm]},
        line width=0.9pt
      }
    ]


    \node[font=\bfseries] at (2,2.15)
      {(a) The graph and coloring $(G,\chi)$};

    \node[pvertex] (p0) at (0,0.75) {$p^0$};
    \node[pvertex] (p1) at (0,-0.75) {$p^1$};

    \node[qvertex] (q0) at (2,0.75) {$q^0$};
    \node[qvertex] (q1) at (2,-0.75) {$q^1$};

    \node[rvertex] (r0) at (4,0.75) {$r^0$};
    \node[rvertex] (r1) at (4,-0.75) {$r^1$};

    \draw[graphedge] (p1) -- (q1);
    \draw[graphedge] (q0) -- (r0);

    \node[
      pairclass,
      fit=(p0)(p1),
      label=above:$P$
    ] {};

    \node[
      pairclass,
      fit=(q0)(q1),
      label=above:$Q$
    ] {};

    \node[
      pairclass,
      fit=(r0)(r1),
      label=above:$R$
    ] {};

    \node[font=\small] at (2,-1.65)
      {$E(G)=\{p^1q^1,\;q^0r^0\}$};

    \node[font=\small] at (2,-2.05)
      {$\chi^{-1}(1)=P,\quad
        \chi^{-1}(2)=Q,\quad
        \chi^{-1}(3)=R$};


    \node[font=\bfseries] at (8.25,2.15)
      {(b) An orientation assignment $\tau$};

    \node[font=\small] at (8.25,1.65)
      {$\tau(P)=1,\qquad \tau(Q)=0,\qquad \tau(R)=1$};

    \node[font=\small] at (6.25,0.95) {$P:$};
    \node[font=\small] at (7.05,0.95) {$p^0$};
    \node[font=\small] at (7.80,0.95) {$\prec_{\tau}$};
    \node[font=\small] at (8.55,0.95) {$p^1$};

    \node[font=\small] at (6.25,0.25) {$Q:$};
    \node[font=\small] at (7.05,0.25) {$q^1$};
    \node[font=\small] at (7.80,0.25) {$\prec_{\tau}$};
    \node[font=\small] at (8.55,0.25) {$q^0$};

    \node[font=\small] at (6.25,-0.45) {$R:$};
    \node[font=\small] at (7.05,-0.45) {$r^0$};
    \node[font=\small] at (7.80,-0.45) {$\prec_{\tau}$};
    \node[font=\small] at (8.55,-0.45) {$r^1$};

    \node[anchor=west,font=\scriptsize] at (9.05,0.95)
      {$\ell_\tau(P)=p^0$};

    \node[anchor=west,font=\scriptsize] at (9.05,0.25)
      {$\ell_\tau(Q)=q^1$};

    \node[anchor=west,font=\scriptsize] at (9.05,-0.45)
      {$\ell_\tau(R)=r^0$};


    \node[font=\bfseries] at (8.25,-1.25)
      {(c) The left digraph $L_\tau(G)$};

    \node[dagnode] (LP) at (5.75,-2.2) {$P$};
    \node[dagnode] (LQ) at (8.25,-2.2) {$Q$};
    \node[dagnode] (LR) at (10.75,-2.2) {$R$};

    \draw[dagedge]
      (LP) --
      node[above,font=\scriptsize]
      {}
      (LQ);

    \draw[dagedge]
      (LQ) --
      node[above,font=\scriptsize]
      {}
      (LR);

  \end{tikzpicture}

  \caption{An example of $L_\tau(G)$ as defined in Section~\ref{sec:support1reduction} on an example $G, \chi$, and $\tau$ is shown in panel~(c). In panel~(b), the notation $u\prec_\tau v$ means that the
    orientation assignment $\tau$ prescribes that a linear order on $V(G)$ should involve $u$ preceding~$v$. The left digraph $L_\tau(G)$, as well as the right digraph $R_\tau(G)$, has its vertices be exactly the pair classes of~$G$, and thus these digraphs depend on~$\chi$ to provide the pair classes of $G$. In $L_\tau(G)$, the connectivity depends on~$\tau$ in that $\tau$ determines what the left endpoints $\ell_\tau$ of the pair classes should be, and whether $P \to Q$ is in $L_\tau(G)$ depends on whether the left endpoint of $Q$ distinguishes the vertices in $P$. $R_\tau(G)$ is defined similarly but with respect to right endpoints instead. In the specific example shown above, $q^1=\ell_\tau(Q)$ distinguishes $p^0$ and~$p^1$, since
    $p^1q^1\in E(G)$ and $p^0q^1\notin E(G)$; hence
    $P\to Q$ is an arc of $L_\tau(G)$. Similarly,
    $r^0=\ell_\tau(R)$ distinguishes $q^0$ and $q^1$, yielding
    $Q\to R$. There are no other arcs in $L_\tau(G)$.
  }
  \label{fig:leftdigraph}
\end{figure}
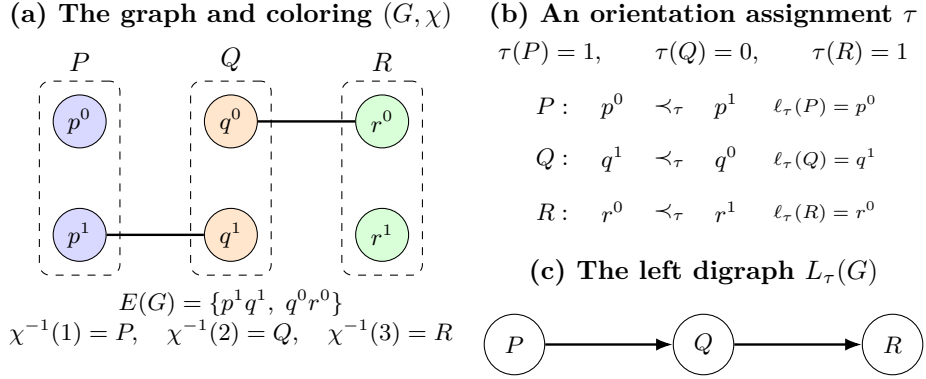

\begin{figure}[t]
\centering
\begin{tikzpicture}[
  vertex/.style={circle,draw,inner sep=1.7pt,minimum size=7mm},
  pair/.style={draw,dashed,rounded corners,inner sep=5pt},
  edge/.style={line width=0.9pt}, scale = 1, transform shape
]
  \node[vertex] (a0) at (0,1) {$a^0$};
  \node[vertex] (a1) at (0,-1) {$a^1$};
  \node[vertex] (b0) at (2,1) {$b^0$};
  \node[vertex] (b1) at (2,-1) {$b^1$};
  \draw[edge] (a1)--(b0);
  \node[pair,fit=(a0)(a1),label=above:$A$] {};
  \node[pair,fit=(b0)(b1),label=above:$B$] {};
  \node at (1,-1.90) {equality gadget $E(A,B)$};

  \node[vertex] (c0) at (5,1) {$a^0$};
  \node[vertex] (c1) at (5,-1) {$a^1$};
  \node[vertex] (d0) at (7,1) {$b^0$};
  \node[vertex] (d1) at (7,-1) {$b^1$};
  \draw[edge] (c1)--(d0);
  \draw[edge] (c1)--(d1);
  \node[pair,fit=(c0)(c1),label=above:$A$] {};
  \node[pair,fit=(d0)(d1),label=above:$B$] {};
  \node at (6,-1.90) {containment gadget $C(A,B)$};
\end{tikzpicture}
\caption{An illustration of the equality and containment gadgets defined in Section~\ref{sec:gadgets}.  The dashed boxes respectively denote pair classes $A$ and $B$.}
\label{fig:wires}
\end{figure}
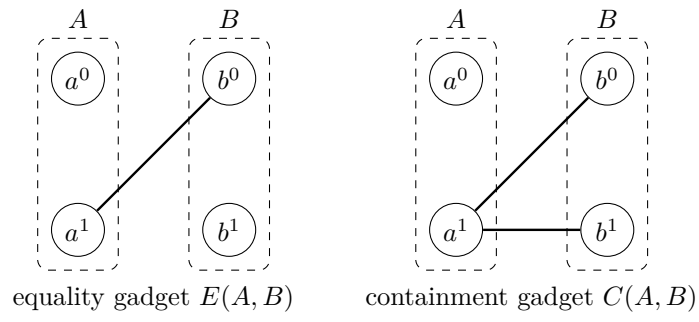


\Cref{fig:Lettericity-Complexity-Venn-diagram} illustrates the status of decision problems concerning lettericity, also when fixing one, two, or even three of the relevant components, i.e., decoder, coloring, or word, that comprise a graph description via a lettering.

\begin{figure}
\centering
\begin{tikzpicture}[venncircle/.style={draw, circle, minimum size=22.5em,
 align=center}, node distance=12em,scale = .7] 
    \node[venncircle] (circle1) {\small \DE\\ \small\textsf{GI}-hard~\cite{GroMorSac2026} \\ \small\textsf{NP}-complete (\Cref{thm:decoder-extension-np-complete})};    
    \node[venncircle, right=-2cm of circle1] (circle2) {\small\CE\\ \small \textsf{NP}-complete (\Cref{thm:coloring-extension-npcomplete})};
    \node (MN) at ($(circle1)!0.5!(circle2)+(0,2em)$){\small $\in$ \textsf{P}~\cite{GroMorSac2026}};
    \node[venncircle, below=3em of MN] (circle3) {\small \WE\\ \small \textsf{NP}-complete (\Cref{thm:word-extension-np-complete})};    
    \node (ML)  at ($(circle1)!0.5!(circle3)$){\small \textsf{GI}-complete~\cite{GroMorSac2026}};  
    \node (NL)  at ($(circle2)!0.5!(circle3)$){\small $\in$ \textsf{P}~\cite{GroMorSac2026}};   
    \node (MNL)  at ($(MN)!0.4!(circle3)$){\small $\in$ \textsf{P}};
    \node (outtoprightnum) [above=1mm of circle2,font=\bfseries\large,text width=4cm] 
     {Prescribed coloring};
    \node (U)[above=1mm of circle1,font=\bfseries\large,text width=4cm]    {Prescribed decoder};
    \node (L)[below=1mm of circle3,font=\bfseries\large,text width=6cm,
 align=center]
     {Prescribed word};
\end{tikzpicture} 
\caption{The outer region, that is, the problem where no part of the solution is prescribed, is referring to \LET, which we proved to be~\textsf{NP}-complete in~\Cref{thm:np-complete}. The encircled region labeled \textbf{Prescribed word} discusses the variant of \LET when the word is fixed, leading to the \textsf{NP}-complete \textsc{Word Extension} problem. If, in addition, the decoder is fixed, we find a \textsf{GI}-complete problem, while fixing the coloring in addition defines a polynomial-time solvable problem. Similarly, the other circles of this Venn diagram can be interpreted.}
\label{fig:Lettericity-Complexity-Venn-diagram}
\end{figure}
\section{Additions to the Conclusions}


Finally, we like to point to the fact that lettericity (and hence thinness) generalize neighborhood diversity. There is another sequence of graph parameters that generalize neighborhood diversity, namely, \emph{uniformicity} and \emph{distinguishing number} (discussed in \cite{Loz2023}), whose complexity status was, to the knowledge of the authors, never discussed. 

Also, the fact that \textsc{Lettericity} is now known to be \textsf{NP}-complete raises the natural question if the natural  decision problem associated with the (smaller) parameter \emph{generalized lettericity} (as introduced in \cite{FenFMRS2025}) is also \textsf{NP}-complete. Notice that this is not a direct consequence of our results as new gadgets would have to be designed as longer words can be part of a decoder.


\end{document}